\documentclass[journal]{IEEEtran}
\usepackage{cite}

\usepackage{graphicx}
\usepackage{subfigure}
\usepackage{epstopdf}
\usepackage{mathrsfs}
\usepackage{pifont}
\usepackage{multirow}
\usepackage{algorithm}
\usepackage{algorithmic}
\usepackage{amsthm,amsmath,amsfonts}
\usepackage{bbm}
\usepackage{mathrsfs}
\newtheorem{lemma}{\textbf{Lemma}}
\usepackage{bm}
\usepackage{subfigure}
\newtheorem{corollary}{Corollary}

\newtheorem{mydef}{Definition}

\usepackage{color}

\begin{document}
%
\title{Multicast eMBB and Bursty URLLC Service Multiplexing in a CoMP-Enabled RAN}
%
%
%

\author{Peng Yang, Xing Xi, Yaru Fu, Tony Q. S. Quek,~\IEEEmembership{Fellow,~IEEE}, Xianbin Cao,~\IEEEmembership{Senior Member,~IEEE},
        Dapeng Wu,~\IEEEmembership{Fellow,~IEEE}
\thanks{
P. Yang, Y. Fu and T. Q. S. Quek are with the Information Systems Technology and Design, Singapore University of Technology and Design, 487372 Singapore.

X. Xi and X. Cao are with the School of Electronic and Information Engineering, Beihang University, Beijing 100083, China, and also with the Key Laboratory of Advanced Technology, Near Space Information System (Beihang University), Ministry of Industry and Information Technology of China, Beijing 100083, China.

D. Wu is with the Department of Electrical and Computer Engineering, University of Florida, Gainesville FL 32611 USA.}
}

\maketitle

\begin{abstract}
This paper is concerned with slicing a radio access network (RAN) for simultaneously serving two typical 5G and beyond use cases, i.e., enhanced mobile broadband (eMBB) and ultra-reliable and low latency communications (URLLC).
Although many researches have been conducted to tackle this issue, few of them have considered the impact of bursty URLLC.
The bursty characteristic of URLLC traffic may significantly increase the difficulty of RAN slicing on the aspect of ensuring a ultra-low packet blocking probability. To reduce the packet blocking probability, we re-visit the structure of physical resource blocks (PRBs) orchestrated for bursty URLLC traffic in the time-frequency plane based on our theoretical results. Meanwhile, we formulate the problem of slicing a RAN enabling coordinated multi-point (CoMP) transmissions for multicast eMBB and bursty URLLC service multiplexing as a multi-timescale optimization problem. The goal of this problem is to maximize multicast eMBB and bursty URLLC slice utilities, subject to physical resource constraints. To mitigate this thorny multi-timescale problem, we transform it into multiple single timescale problems by exploring the fundamental principle of a sample average approximation (SAA) technique. Next, an iterative algorithm with provable performance guarantees is developed to obtain solutions to these single timescale problems and aggregate the obtained solutions into those of the multi-timescale problem. We also design a prototype for the CoMP-enabled RAN slicing system incorporating with multicast eMBB and bursty URLLC traffic and compare the proposed iterative algorithm with the state-of-the-art algorithm to verify the effectiveness of the algorithm.
\end{abstract}


%
\IEEEpeerreviewmaketitle

\section{Introduction}
%
%
%
%
\IEEEPARstart{5}{G} and emerging 6G wireless networks are envisioned to accommodate different service requirements concerning throughput, latency, reliability, availability and operational requirements as well, e.g., energy efficiency and cost efficiency \cite{rost2017network}.
These service requirements are proposed by mobile networks and some novel and significant application areas such as Industry 4.0, vehicular communication, and smart grid. Owing to the huge market prospects for these application areas, the International Telecommunication Union (ITU) has categorized the services proposed by them into three major use cases: enhanced mobile broadband (eMBB) including ultra-high definition (UHD) TV, massive machine-type communications (mMTC) for metering, logistics, smart agriculture, and ultra-reliable and low latency communications (URLLC) for autonomous driving and automated factory \cite{series2015imt}.
Further, in order to provide cost-efficient solutions, it is agreed by some telecommunication organizations including the Third Generation
Partnership Project (3GPP) and the Next Generation Mobile
Network (NGMN) Alliance on the convergence of each use case onto a common physical infrastructure instead of deploying individual network solution for each use case.

To achieve the goal of converging all use cases onto a shared infrastructure, the concept of network slicing has been proposed. The key idea of network slicing is to logically isolate network resources and functions customized for specific requirements of a common physical infrastructure \cite{alliance2015ngmn}. A collection of logically isolated core network (CN) and radio access network (RAN) functions is deemed as a network slice. Most recent researches on network slicing focus on slicing the CN and the RAN.
Slicing the CN affects both functionalities of the control plane such as mobility management, session management, and authentication (as hosted in mobility management entities and home subscribe
servers), and functionalities of the user plane (UP) (e.g., those in the serving gateway and packet data network gateway), both of which become programmable and auto-configurable.
The research of RAN slicing is still in its infancy and is highly challenging. This is due to the sharing characteristic of radio resources and complicate parameter configurations, e.g., the design of the time-frequency plane, round trip time (RTT), transmission time interval (TTI), and hybrid automatic repeat request (HARQ) options \cite{campolo20175g}. This paper aims at studying the issue of slicing RAN resources to provide better network resource isolation and increase statistical multiplexing.

\subsection{Prior Work}
During the past few years, plenty of works on implementing resource isolation and improving statistical multiplexing had been contributed to the RAN slicing research community.
For example, the works in \cite{albonda2019efficient, alsenwi2019embb, guo2019enabling,yang2020ran,li2019side,shi2018hierarchical,kurtz2018network,chien2019end,ferrus20185g,salhab2018optimization} exploited the performance of service multiplexing of RAN slicing on the aspect of radio resource optimization. Based on an off-line reinforcement learning and a low complexity heuristic algorithm, the work in \cite{albonda2019efficient} proposed to allocate RAN resources to eMBB and vehicle-to-everything (V2X) slices such that the resource utilization was maximized and the quality of service (QoS) requirements of eMBB and V2X slices were fulfilled.
The work in \cite{alsenwi2019embb} proposed a risk-sensitive based formulation to allocate network resources to incoming URLLC traffic while minimizing the risk (low data rate) of eMBB transmission and ensuring the reliability of URLLC transmission. Additionally, the work in \cite{guo2019enabling} developed a novel slice scheduling framework to enable 5G RAN slicing for eMBB, URLLC, and mMTC service multiplexing.

Different from the above works \cite{albonda2019efficient, alsenwi2019embb, guo2019enabling,yang2020ran,li2019side,shi2018hierarchical,kurtz2018network,chien2019end,ferrus20185g,salhab2018optimization} considering orthogonal spectrum resources, some other works \cite{matera2018non,park2018urllc,popovski20185g} had discussed the nonorthogonal case to improve the spectrum utilization.
From an information-theoretic perspective, the work in \cite{matera2018non} discussed the performance of a cloud RAN architecture of serving URLLC and eMBB traffic in a non-orthogonal multiple access (NOMA) manner.
The work in \cite{park2018urllc} investigated methods of non-orthogonal eMBB and URLLC slicing to support virtual reality over wireless cellular systems.
Besides, the work in \cite{popovski20185g} researched the advantages of allowing for non-orthogonal sharing of RAN resources in uplink communications from a group of eMBB, URLLC, and mMTC devices to a base station (BS).

Except for the RAN resource optimization, some researches designed architectures of implementing RAN slicing for service multiplexing. For instance,
the work in \cite{adamuz2019harmonizing} proposed a description model to enable the translation of RAN slice requirements (i.e., low latency, high reliability, high throughput, and massive device support) into customized virtualized radio functionalities defined through network function virtualization descriptors.

\subsection{Motivation and Contribution}
Differ from previous research efforts on enabling flexible and scalable RAN slicing for statistical multiplexing, this paper explores the cost efficiency issue of slicing the radio resource shared RAN for eMBB and URLLC service multiplexing.

This research topic is quite challenging as it has to mitigate at least three tricky issues \cite{tang2019service}: \emph{1) Two timescales issue:} RAN slicing is expected to be executed in a timescale of minutes to hours so as to keep in pace with the timescale of slicing upper layers. Nevertheless, wireless channel changes in a timescale of milliseconds, which is much shorter than the duration of a slice operation. As a result, how to tackle the two timescales issue is challenging;
\emph{2) Isolation of inter-slice interference:} different types of slices (i.e., eMBB slices and URLLC slices in this paper) in a RAN slicing system share the common physical channel; thus, how to isolate inter-slice interference is a big challenge;
\emph{3) Maximization of total utility:} different types of slices in a RAN slicing system also share common radio resources; thus, how to efficiently orchestrate resources for diverse slices such that the total system utility can be maximized is difficult.

A recent work \cite{tang2019service} proposed a C-RAN slicing architecture for eMBB and URLLC service multiplexing to deal with the above challenging issues. Specifically, it first utilized an alternating direction method of multipliers (ADMM) \cite{boyd2011distributed} joint with a sample average approximation (SAA) technique \cite{kim2015guide} to tackle the two timescales issue via simplifying the two timescales problem into multiple single timescale problems. Then
a flexible frequency division duplex (FDD) technique was exploited to orthogonalize diverse slices to isolated inter-slices. At last, it designed a generic utility framework that maximized the total utility of eMBB and URLLC service multiplexing through efficiently admitting eMBB and URLLC slice requests.

However, the work \cite{tang2019service} assumed that URLLC traffic was uninterruptedly generated and ignored the significant bursty characteristic of URLLC traffic \cite{hou2018burstiness}. The bursty URLLC traffic may further exacerbate the solving difficulty of slicing the RAN for URLLC involved service multiplexing from the following three perspectives:
\begin{itemize}
    \item \textbf{Time-frequency plane design:} except for the challenging radio resource allocation for total utility maximization, the transmission of bursty URLLC traffic needs an efficient design of physical resource blocks (PRBs) in the time-frequency plane to reduce the packet blocking probability;
    \item \textbf{Resource utilization issue:} one of the efficient proposals in future wireless communication networks to handle the uncertainty (including bursty) is to reserve network resources, which may waste a large amount of resources. Therefore, improving resource utilization is non-trivial for the bursty URLLC service provision;
    \item \textbf{Slice demand and slice supply mismatch:} bursty URLLC packets need to be immediately scheduled (slice demand) if there are available resources and the system utility can be maximized. However, the slice creation in the RAN slicing system (slice supply) is time costly.
\end{itemize}

This paper investigates the coordinated multi-point (CoMP) enabled RAN slicing for multicast eMBB and bursty URLLC service multiplexing, and the main contributions of this paper can be summarized as follows:
\begin{itemize}
    \item Guided by theoretical results, we re-visit the time-frequency structure of PRBs orchestrated for bursty URLLC transmission to reduce the URLLC packet blocking probability.
    \item A concept of 'slice of subslices' followed by a resource mask algorithm are developed to mitigate the mismatch issue of URLLC slice demand and supply and improve the resource utilization as well.
    \item We define a multicast eMBB slice utility function and a bursty URLLC slice utility function reflecting parameters of eMBB and URLLC slice requests. We formulate the CoMP-enabled RAN slicing problem for multicast eMBB and bursty URLLC service multiplexing as a multi-timescale optimization problem with a goal of maximizing eMBB and URLLC slice utilities, subject to constraints on total system bandwidth and transmit power.
    \item Based on the fundamental principle of an SAA technique, we transform the multi-timescale problem into multiple mixed-integer positive semidefinite programming (MISDP) problems of the single timescale. An iterative algorithm, which is proven to be convergent, is developed to achieve solutions to these single timescale problems. Besides, in this algorithm, an ADMM method followed by a new restoration scheme with provable performance guarantees are exploited to aggregate the achieved single timescale solutions into multi-timescale ones.
    \item We also design a prototype for the CoMP-enabled RAN slicing system and conduct plenty of simulations to verify the effectiveness of the iterative algorithm.
\end{itemize}

\subsection{Organization}
The remainder of this paper is organized as follows: Section \uppercase\expandafter{\romannumeral2} builds the system model. Based on the model, a RAN slicing problem for multicast eMBB and bursty URLLC service multiplexing is formulated in Section \uppercase\expandafter{\romannumeral3}. Section \uppercase\expandafter{\romannumeral4} aims to transform the formulated problem. Section \uppercase\expandafter{\romannumeral5} and Section \uppercase\expandafter{\romannumeral6} propose to mitigate the transformed problem with system generated channel coefficients and sensed channel coefficients, respectively. In Section \uppercase\expandafter{\romannumeral7}, we design a RAN slicing system prototype. The simulation is conducted in Section \uppercase\expandafter{\romannumeral8}, and Section \uppercase\expandafter{\romannumeral9} concludes this paper.

\emph{Notation:} Boldface uppercase letters denote matrices,
whereas boldface lowercase letters denote vectors. The
superscripts $(\cdot)^{\rm T}$ and $(\cdot)^{\rm H}$ denote transpose and conjugate transpose matrix operators. ${\rm tr}(\cdot)$, ${\rm rank}(\cdot)$, $|\cdot|$ and $\lceil \cdot \rceil$ denote the trace,
the rank, the absolute value, and the rounding up operators,
respectively. By ${\bm X} \succeq 0$ we denote that $\bm X$ is a Hermitian positive-semidefinite matrix.

\section{System Model}
We consider a CoMP-enabled RAN slicing system for multicast eMBB and bursty URLLC multiplexing service provision. In this system, there are a number of $N^e$ eMBB ground user equipments (UEs), a number of $N^u$ URLLC ground UEs and $J$ BSs. All ground UEs are assumed to be spatially distributed in a restricted geographical area $\mathbb R$ according to a random distribution $\Phi$, and the BSs are assumed to be regularly distributed at the boundary of $\mathbb R$. Each BS is equipped with $K$ antennas, and each UE is equipped with a single receive antenna.
The $J$ BSs are connected by fiber to realize UE data sharing and time-frequency synchronization. Each UE will receive data transmitted by all BSs through physical downlink sharing channels (PDSCH) and will coherently merge the received data such that the inter-user interference can be suppressed.

Besides, the time of the system is discretized and is composed of two different timescales, i.e., time slot and minislot. At the beginning of each time slot, a software-defined RAN coordinator (SDRAN-C) in the system will decide whether to accept or reject received network slice requests reflecting profiles (such as the number and QoS requirements) of eMBB and URLLC UEs. If a slice request is accepted, the system will be reconfigured, which is time costly and usually in a timescale of minutes to hours, and system resources will be re-orchestrated to accommodate the slice requirement.
At the beginning of each minislot, BSs will generate beamformers matching time-varying channels.
We assume that each time slot can be divided into $T$ minislots and discuss two kinds of network slices, that is, multicast eMBB slices and Unicast URLLC slices. The collection of eMBB slices is denoted by ${\cal S}^e := \{1, 2, \ldots, {S^e}\}$, and the set of URLLC slices is denoted by ${\cal S}^u := \{1, 2, \ldots, {S}^u\}$.

\subsection{Multicast eMBB Slice Model}
According to the definition of a network slice (especially from the perspective of the QoS requirement of a slice), an eMBB network slice request can be defined as follows:
\begin{mydef}
\rm {For any multicast eMBB slice $s \in {\cal S}^e$, its network slice request is composed of two components \cite{tang2019service}:
\begin{itemize}
    \item \textbf{The number of eMBB UEs:} the symbol $I_s^e$ is utilized to represent the number of eMBB UEs grouped into the slice $s$.
    \item \textbf{QoS requirements of eMBB UEs:} eMBB UEs prefer high throughput and greater network capacity. As a result, we characterize the QoS requirements of eMBB UEs as their minimum data rates, denoted by $R_s$.
\end{itemize}
To this end, the tuple $\{I_s^e, R_s\}$ is used to represent an eMBB slice request of $s$.
}
\end{mydef}

\textbf{Remark:} all UEs, the set of which is denoted by ${\cal I}_s^e$, in an eMBB slice $s$ have the same QoS requirement. A slice request of $s$ can be accepted only if the QoS requirements of all UEs in $s$ are accommodated. A binary variable $b_{s}^e \in \{0, 1\}$ is then utilized to indicate whether the slice request of $s$ is accepted by the SDRAN-C. $b_s^e = 1$ denotes that the slice request is accepted; otherwise $b_s^e = 0$.

eMBB UEs may experience annoying inter-slice interference if two or more eMBB slice requests are accepted at the same time.
This type of interference may significantly degrade the QoS of eMBB UEs. Like the work in \cite{tang2019service}, a flexible frequency division multiple access (FDMA) technique is leveraged to mitigate the inter-slice interference and enhance the QoS experienced by eMBB UEs. In this technique, the resource block in the frequency plane assigned to each activated UE (i.e., a UE belonging to an accepted slice) can be tailored without violating the constraint on the total amount of spectrum. On the other hand, owing to the exploration of a coherent transmission and merge technique, the intra-slice interference can be effectively suppressed.

Next, we denote a beamformer pointing to all eMBB UEs in $s$ ($s \in {\cal S}^e$) transmitted by the $j$-th BS ($j \in {\cal J}$) at minislot $t$ by ${\bm v}_{j,s}(t) \in {\mathbb C}^K$. The time-varying channel between the $j$-th BS and the $i$-th UE ($i \in {\cal I}_s^e$) in $s$ at minislot $t$ is denoted as ${\bm h}_{ij,s}(t) \in {\mathbb C}^K$. Suppose that ${\bm h}_{ij,s}(t)$ is subject to a random distribution $\Phi_{\bm h}$ that is imperfectly known by the SDRAN-C.
Meanwhile, for any $i$, $j$, and $s$, the random variable ${\bm h}_{ij,s}(t)$ at each $t$ is assumed to be independent and identically distributed (i.i.d).

For all multicast UEs in $s$, let $u_s^e(t)$ be the sharing signal of them at $t$ with ${\mathbb E}[|u_s^e(t)|^2] = 1$. Thus, under the setup of the CoMP downlink transmission, the received signal ${\hat u}_{i,s}^e(t)$ of a multicast UE $i$ can be expressed as
\begin{equation}\label{eq:eMBB_receiving_signal}
    {\hat u}_{i,s}^e(t) = \sum\limits_{j \in {\cal J}} {{\bm h}_{ij,s}^{\rm H}}{{(t)}}{{\bm v}_{j,s}(t)u_s^e(t)}  + {\delta _{i,s}}(t), \forall i \in {\cal I}_s^e, s \in {\cal S}^e
\end{equation}
where the first term on the right-hand side (RHS) represents the desired signal by UE $i$ in $s$ and the second term $\delta_{i,s}(t) \sim {\cal CN}(0, \sigma_{i,s}^2) $ denotes the additive white Gaussian noise (AWGN) received at $i$. The corresponding SNR experienced by $i$ in $s$ at minislot $t$ can be written as
\begin{equation}\label{eq:eMBB_snr}
    SNR_{i,s}^e(t) = \frac{{|\sum\nolimits_{j \in {\cal J}} {{{\bm h}_{ij,s}^{\rm H}}{{(t)}}{{\bm v}_{j,s}}(t)} {|^2}}}{{\phi \sigma _{i,s}^2}}, \forall i \in {\cal I}_s^e, s \in {\cal S}^e
\end{equation}
where $\phi > 1$ represents the SNR loss due to imperfect channel state information (CSI) sensing at the receiver \cite{liu2014energy}.

With the mathematical expression of SNR, the achievable data rate $\gamma _{i,s}^e(t)$ of UE $i$ in $s$ at minislot $t$ can take the following form according to the Shannon formula
\begin{equation}\label{eq:eMBB_throught}
    \gamma _{i,s}^e(t) = \omega_s^e(\bar t){\log _2}(1 + SNR_{i,s}^e(t)), \forall i \in {\cal I}_s^e, s \in {\cal S}^e
\end{equation}
where $\omega_{s}^e(\bar t)$ denotes the system bandwidth allocated to $s$ at time slot $\bar t$.

Owing to the channel deep fading, active eMBB UEs may experience signal outage. We, therefore, model the necessary condition for the SDRAN-C to accept the request of slice $s$ as
\begin{equation}\label{eq:eMBB_QoS}
    {\rm Pr}(\gamma _{i,s}^e(t) \ge R_s) \ge 1 - \epsilon, \forall i \in {\cal I}_s^e, s \in {\cal S}^e
\end{equation}
where $\epsilon \in (0,1)$ is the maximum tolerable system outage probability.

\subsection{Bursty URLLC Slice Model}
Different from eMBB UEs in terms of QoS requirements, URLLC UEs need to successfully transmit and decode data packets with extremely low latency (1 ms) and extremely high reliability ($99.999\%$). Thus, the definition of a bursty URLLC slice request can be described as follows.
\begin{mydef}
\rm{For any bursty URLLC slice $s \in {\cal S}^u$, its network slice request is composed of four components:
\begin{itemize}
    \item \textbf{The number of URLLC UEs:} the symbol $I_s^u$ represents the number of URLLC UEs classified into the slice $s$.
    \item \textbf{QoS requirements of URLLC UEs:} URLLC UEs require low latency end-to-end transmissions, which is significantly different from eMBB UEs. Thus, the communication latency $D_s$ is leveraged to characterize the QoS requirements of URLLC UEs.
    \item \textbf{Codeword error decoding probability:} $\alpha$ that should not be greater than a threshold is used to represent the error probability of decoding a URLLC codeword\footnote{A URLLC packet will usually be coded before transmission, and the generated codeword will be transmitted in the air interface such that the transmission reliability can be improved.}.
    \item \textbf{Packet blocking probability:} $\beta$ that should be lower than a threshold is utilized to denote the URLLC packet blocking probability.
\end{itemize}
In this way, a four tuples $\{I_s^u, D_s, \alpha, \beta\}$ can be involved to represent a bursty URLLC slice request of $s$.
}
\end{mydef}

\textbf{Remark:} In a bursty URLLC slice request, $\alpha$ and $\beta$ are jointly utilized to characterize the reliability requirement of URLLC transmission. All UEs, the set of which is denoted by ${\cal I}_s^u$, in $s$ have the same communication latency requirement. Besides, a variable $b_s^u \in \{0, 1\}$ is introduced to indicate whether the request of $s$ can be accepted. If yes, we set $b_s^u = 1$; otherwise, $b_s^u = 0$.

As mentioned above, a slice cannot be immediately created even if it is accepted as the construction process of a network slice is time costly. However, URLLC packets have stringent ultra-low latency requirements. Once arrived, URLLC packets should be immediately scheduled and transmitted. Therefore, the slice creation with the scale of a time slot may be inappropriate for URLLC service. We then propose the following 'slice of subslices' concept to tackle this issue.

\subsubsection{Slice of subslices}
In the concept of 'slice of subslices', a URLLC slice (or called global slice) is virtually partitioned into multiple (local) subslices.
Local subslices may evolve individually for the timely and flexible URLLC service provision. There is no need to rebuild the global slice, which is time-consuming when the local evolution is executed.

Guided by the concept, we propose a subslice resource mask scheme under a crucial assumption.

\emph{Assumption:} (\textbf{Always request acceptance}) The SDRAN-C always accepts all the URLLC slice requests at each time slot if there is spare bandwidth. Otherwise, some URLLC slice requests will be declined.

Since the arrival process of URLLC packets has the sporadic and bursty characteristics \cite{liu2014energy}, we preferentially allocate bandwidth to eMBB slices and then reserve the spare bandwidth for URLLC slices at each time slot. Together with the stringent latency requirements of URLLC packets, this kind of request acceptance assumption is reasonable.

In principle, the resource mask scheme defines the resource mask as a vector indicating a batch of PRBs dynamically assigned to each subslice. The SDRAN-C is responsible for adjusting the resource mask via monitoring the channel, which achieves the adaptive resource allocation among URLLC subslices according to the changing network dynamics.
With a slight abuse of notation, for each URLLC slice $s \in {\cal S}^u$, we use the vector ${\bm b}_s^u(t) = [b_{1,s}^u(t), b_{2,s}^u(t), \ldots, b_{{I_s^u},s}^u(t)]$ to represent a resource mask, where $b_{i,s}^u \in \{0, 1\}$ for all $i \in {\cal I}_s^u$ denotes whether the SDRAN-C will allocate a batch of PRBs to the subslice corresponding to UE $i$. $b_{i,s}^u = 1$ if some PRBs are allocated; otherwise, $b_{i,s}^u = 0$.

Then, it is essential to represent URLLC subslice requests. Based on the definition of a URLLC slice request, we use the triple $\{D_s, \alpha, \beta\}$ to represent the request of a URLLC subslice. In this triple, the number of URLLC UEs in a subslice is not included as Unicast subslices are considered, i.e., there is only one UE in each subslice.

Although the advantages of dynamical PRB allocation are attractive, it is highly challenging to allocate an appropriate amount of PRBs (in both time and frequency planes) to a URLLC subslice. This is because
\begin{itemize}
    \item URLLC packets have a stringent low latency requirement; resources in the time plane allocated to URLLC packets cannot exceed the maximum packet latency.
    \item Even systems with a great bandwidth configuration may occasionally suffer from packet congestion owing to the stochastic variations in the packet arrival process, and occasionally, there may not be enough spare bandwidth to transmit a new URLLC packet simultaneously \cite{anand2018resource}.
\end{itemize}

We next study the efficient time-frequency structure design of PRB for URLLC transmission.

\subsubsection{Structure design of PRB for URLLC transmission}
Define a vector of URLLC packet arrival rates ${\bm \lambda} = ({\bm \lambda}_1, \ldots, {\bm \lambda}_s, \ldots, {\bm \lambda}_{S^u})$ with ${\bm \lambda}_s = (\lambda_{1,s},\ldots,\lambda_{I_{s}^u,s})$ representing the arrival rates of URLLC UEs in ${\cal I}_s^u$. A URLLC packet destined to $i \in {\cal I}_s^u$ is allocated with a bandwidth of $\omega_{i,s}^{u}$ for a period of time $d_{i,s}$. These values are related to the channel use $r_{i,s}^{u}$ by $\kappa \omega_{i,s}^{u}d_{i,s} = \lambda_{i,s}r_{i,s}^{u}$, where $\kappa $ is a constant denoting the number of channel uses per unit time per unit bandwidth of the FDMA frame structure and numerology. Since URLLC packets destined to $s$ have a deadline of $D_s$ seconds, we shall always choose $d_{i,s} \le D_s$\footnote{In this work, we consider the circumstance of one time transmission. However, in order to further improve the reliability of transmitting URLLC packets, some HARQ schemes, which focus on the performance analysis of retransmission, are deserved to be exploited in the future.}. For ease of analysis, we assume that for all UEs in $s$, $D_s$ is an integer multiple of $d_{i,s}$. Thus, the following vectors may be enough to characterize a scenario of URLLC service provision: the set of packet transmission latency for all UEs in $s$ is represented as ${\bm d} = \{{\bm d}_s\}$ with ${\bm d}_s = \{d_{1,s},\ldots,d_{I_s^u,s}\}$.
The channel use set for all UEs in $s$ is denoted as ${\bm r^u} = \{{\bm r}_s^u\}$ with ${\bm r}_s^u = \{r_{1,s}^u,\ldots,r_{I_s^u,s}^u\}$.
The correspondingly allocated bandwidth set is ${\bm \omega}^u = \{ {\bm \omega}_s^u \}$, where ${\bm \omega}_s^u = \{\omega_{1,s}^u,\ldots,\omega_{I_s^u,s}^u\}$, and the average packet arrival rates of all UEs in $s$ in a certain duration is denoted as $\bm{\rho} = \{\rho_1, \ldots, \rho_{S^u}\}$, where $\rho_s = \lambda_{i,s} d_{i,s}$.

On the one hand, shortening the packet transmission latency implies that fewer PRBs are available in the frequency plane, a fact that will definitely cause more queueing effect and significantly increase the blocking probability of a URLLC packet. On the other hand, narrowing a PRB in the frequency domain implies more concurrent transmissions, which is beneficial for decreasing the blocking probability of a URLLC packet. It will, however, incur high packet transmission latency.

Therefore, the following issue should be addressed when designing the structure of a PRB for URLLC transmission: \emph{how to tailor PRBs in the time-frequency plane to reduce the blocking probability of a URLLC packet}?

For each slice $s \in {\cal S}^u$, let $p_s({\bm \omega}^u, {\bm d}, {\bm \lambda}, W^u)$ denote the blocking probability experienced by an arrival packet destined to a UE in $s$. The following Lemma provides us with a crucial clue on the time-frequency resource orchestration for URLLC packet transmission.

\begin{lemma}\label{lem:probability}
For a given ${\bm \omega}^u$, ${\bm d}$, and a positive integer $q$, define ${\hat {\bm \omega}^u} = ({\bm \omega}_1^u, \ldots, {\bm \omega}_s^u/q, \ldots, {\bm \omega}_{S^u}^u)$ and ${\hat {\bm d}} = ({\bm d}_1, \ldots, q{\bm d}_s, \ldots, {\bm d}_{S^u})$. Under the case of one time transmission, if $\rho_s < 1$, then for a great system bandwidth $W^u$, we have $p_s({\bm \omega}^u, {\bm d}, {\bm \lambda}, W^u) \ge p_s({\hat {\bm \omega}^u}, {\hat {\bm d}}, {\bm \lambda}, W^u)$.
\end{lemma}

\begin{proof}
Please refer to Appendix A.
\end{proof}

This Lemma shows that a system with narrowed PRBs in the frequency plane not only increases the number of concurrent transmissions of URLLC packets destined to UEs in $s$ but also is beneficial for UEs in other slices.

Therefore, one should scale $d_{i,s}$ with an integer $q$ such that $qd_{i,s} = D_s$. This motivates us to choose $d_{i,s}$ and $\omega_{i,s}^u(t)$ of a URLLC subslice as follows.
\begin{equation}\label{eq:omega_s}
   d_{{i,s}}(t)  = D_s \ {\rm and} \ \omega_{{i,s}}^{u}(t) = \frac{\lambda_{i,s}{r_{{i,s}}^{u}(t)}}{{\kappa D_s}}, {\forall i \in {\cal I}_s^u}, s \in {\cal S}^u
\end{equation}

Besides, considering the reliability requirement of transmitting URLLC packets, the bandwidth $W^{u}$ reserved for all URLLC slices should satisfy a certain condition relating to the blocking probability $\beta$ of a URLLC packet. To this aim, a multi-class extension of the classical square-root staffing rule (see \cite{harchol2013performance} for more details) to correlate $W^u$, $\bm r^u$, $\bm \lambda$, and $\beta$ is exploited. Particularly, to provide communication services for URLLC traffic of $\bm \lambda$ with reliability $\beta$ for a given $\bm r^u$, the mathematical expression of the minimum reserved bandwidth can take the following form \cite{anand2018resource}
\begin{equation}\label{eq:URLLC_bandwidth}
    W^u(t) = \varsigma^{\rm mean}(\bm r^u(t)) + Q^{-1}(\beta)\sqrt{\varsigma^{\rm variance }\bm r^u(t)}
\end{equation}
where $\varsigma^{\rm mean}(\bm r^u(t)) = \sum\nolimits_{s \in {\cal S}^u } {\sum\nolimits_{i \in {\cal I}_s^u} {b_{i,s}^u(t){\lambda _{i,s}}\frac{{{r_{i,s}^u(t)}}}{\kappa }} } $ is the mean of the required system bandwidth, and $\varsigma^{\rm variance }(\bm r^u(t)) =  \sum\nolimits_{s \in {\cal S}^u} {\sum\nolimits_{i \in {\cal I}_s^u} {b_{i,s}^u(t){\lambda _{i,s}}\frac{{r_{i,s}^{2u}(t)}}{{{\kappa ^2}{D_s}}}} } $ is the variance of the required bandwidth.

In (\ref{eq:URLLC_bandwidth}), the reserved bandwidth $W^u(t)$ is related to channel uses of URLLC UEs. We, therefore, discuss how to model channel uses in the following subsection.

\subsubsection{Channel uses of URLLC UEs}
For a URLLC slice $s \in {\cal S}^u$, let $u_{i,s}^u(t)$ be the data symbol destined to URLLC UE $i$ for all $i \in {\cal I}_s^u$ during minislot $t$ with ${\mathbb E}[|u_{i,s}^u(t)|^2] = 1$, and ${\bm g}_{ij,s}(t) \in {\mathbb C}^K$ be the transmit beamformer pointed at $i$ from BS $j$ at $t$. Just like \cite{tang2019service}, an FDMA scheme is applied to $s$ to alleviate inter-subslice interference. The received signal ${\hat u}_{i,s}^u(t)$ at UE $i$ in $s$ during minislot $t$ can then take the following form
\begin{equation}\label{eq:eMBB_receiving_signal}
    {\hat u}_{i,s}^u(t) = \sum\limits_{j \in {\cal J}} {{\bm h}_{ij,s}^{\rm H}}{{(t)}}{{\bm g}_{ij,s}(t)u_{i,s}^u(t)}  + {\delta _{i,s}}(t), \forall i \in {\cal I}_s^u, s \in {\cal S}^u
\end{equation}
where the first term on the RHS denotes the desired signal for UE $i$. The corresponding SNR received at $i$ in $s$ over minislot $t$ can be expressed as
\begin{equation}\label{eq:URLLC_snr}
    SNR_{i,s}^u(t) = \frac{{|\sum\nolimits_{j \in {\cal J}} {{{\bm h}_{ij,s}^{\rm H}}{{(t)}}{{\bm g}_{ij,s}}(t)} {|^2}}}{{\phi \sigma _{i,s}^2}}, \forall i \in {\cal I}_s^u, s \in {\cal S}^u
\end{equation}

Owing to the stringent low latency requirement, the length of a URLLC packet is typically very short. As a result, the achievable data rate and error probability of packet transmission cannot be effectively captured by Shannon's capacity formula. Instead, the receiving data rates of URLLC packets may fall into a finite blocklength channel coding regime, which are derived in \cite{yang2014quasi}. Mathematically, in an AWGN channel, the number of information bits $L_{i,s}^u(t)$ for all $i \in {\cal I}_s^u$ and $s \in {\cal S}^u$ at $t$ that is transmitted with a codeword decoding error probability of $\alpha$ and $r$ channel uses can be accurately approximated by
\begin{equation}\label{eq:URLLC_bit_length}
    \begin{array}{l}
L_{i,s}^u(t) \approx r_{i,s}^u(t)C(SNR_{i,s}^u(t)) - \\
\qquad {Q^{ - 1}}(\alpha )\sqrt {r_{i,s}^u(t)V(SNR_{i,s}^u(t))}, \forall i \in {\cal I}_s^u, s \in {\cal S}^u
\end{array}
\end{equation}
where $C(SNR_{i,s}^u(t)) = \log_2(1 + SNR_{i,s}^u(t) )$ is the AWGN channel capacity per Hz under the infinite blocklength assumption, $V(SNR_{i,s}^u(t)) = \ln^2 2\left( {1 - \frac{1}{{{{(1 + SNR_{i,s}^u(t))}^2}}}} \right)$ denotes the channel dispersion.

The expression of $L_{i,s}^u(t)$ is complicate, which significantly hinders the theoretical derivation of the optimization problem formulated in the following section. To handle this issue, we approximate $L_{i,s}^u(t)$ in two cases:
\begin{itemize}
\item \emph{Case I: enforced SNR constraint.} When the received SNR at a URLLC UE is not less than 5 dB, which is easily achieved in CoMP transmission networks (especially when supporting URLLC), $V(SNR_{i,s}^u(t))$ can be accurately approximated as $\ln^2 2$ \cite{schiessl2015delay}.
\item \emph{Case II: relaxed SNR constraint.} A key observation is that $V(SNR_{i,s}^u(t)) < \ln^2 2$ when the received SNR is smaller than 5 dB. Thus, through substituting $V(SNR_{i,s}^u(t)) = \ln^2 2$ into (\ref{eq:URLLC_bit_length}), we can achieve the lower bound of $L_{i,s}^u(t)$. If the lower-bounded value is applied to optimize the allocation of resources, the bandwidth constraint depicted in the next section can be satisfied.
\end{itemize}

Accordingly, we can further approximate $L_{i,s}^u(t)$ as $r_{i,s}^u(t)C(SNR_{i,s}^u(t)) - {Q^{ - 1}}(\alpha )\sqrt {r_{i,s}^u(t)}$. With this approximated $L_{i,s}^u(t)$, we can write $r_{i,s}^u(t)$ as a function of $\alpha$ with
\begin{equation}\label{eq:URLLC_channel_use}
    \begin{array}{l}
r_{i,s}^u(t) = \frac{{L_{i,s}^u(t)}}{{C(SNR_{i,s}^u(t))}} + \frac{{{{{Q^{ - 2}}(\alpha)}}}}{{2{{(C(SNR_{i,s}^u(t)))}^2}}} + \\
 \ \frac{{{{{Q^{ - 2}}(\alpha)}}}}{{2{{(C(SNR_{i,s}^u(t)))}^2}}}\sqrt {1 + \frac{{4L_{i,s}^u(t)C(SNR_{i,s}^u(t))}}{{{{{Q^{ - 2}}(\alpha)}}}}}, \forall i \in {\cal I}_s^u, s \in {\cal S}^u
\end{array}
\end{equation}

\begin{proof}
If we substitute $\sqrt{r_{i,s}^u(t)} = x$, then (\ref{eq:URLLC_bit_length}) is a quadratic equation in $x$. Solving it we can achieve the closed-form expression for $r_{i,s}^u(t)$ in (\ref{eq:URLLC_channel_use}).
\end{proof}

\section{Problem formulation}
Based on the above system model, this section aims to formulate the problem of CoMP-enabled RAN slicing for multicast eMBB and bursty URLLC multiplexing service provision.

\subsection{Inter-Slice Constraints}
As the CoMP transmission mode is applied to the multicast eMBB and bursty URLLC multiplexing service provision, and each BS has the maximum transmit power $E_j$ for all $j \in {\cal J}$, we can write the power consumption constraint of each BS as
\begin{equation}\label{eq:RRH_energy}
   \sum\limits_{s \in {{\cal S}^e}} {b_s^e(\bar t){{\bm v}_{j,s}^{\rm H}}{{(t)}}{{\bm v}_{j,s}}(t)}  + \sum\limits_{s \in {\cal S}^u} {\sum\limits_{i \in {\cal I}_s^u} {b_{i,s}^u(t){{\bm g}_{ij,s}^{\rm H}}{{(t)}}{{\bm g}_{ij,s}}(t)} }  \le {E_j}
\end{equation}
where the first term on the left-hand side (LHS) denotes the power consumption of the $j$-th BS for multicasting signals to eMBB UEs at minislot $t$, the second term on the LHS represents the power consumption of the $j$-th BS for unicasting signals to URLLC UEs at minislot $t$.

Since the multicast service for eMBB UEs is considered, and bandwidths allocated to eMBB slices and bursty URLLC slices are orthogonal, the system bandwidth constraint can be given by
\begin{equation}\label{eq:total_bandwidth}
   \sum\limits_{s \in {{\cal S}^e}} {b_s^e(\bar t)\omega_s^e(\bar t)}  + W^u(t)  \le W
\end{equation}
where $W$ denotes the maximum system bandwidth.

\subsection{Utility Function Design}
The goal of the SDRAN-C in the system is to maximize the achieved total utility in a period of time, which is composed of the achieved utility for multicast eMBB service provision and the utility for bursty URLLC service provision. In this paper, we leverage the energy efficiency that is popularly exploited in resource allocation problems to model the achieved utility.

As each time slot is independent of each other over the whole time slots, rather than modelling the achieved utility over the whole time slots, we study the achieved utility during a randomly selected time slot $\bar t$. Besides, during the slot $\bar t$ consisting of $T$ minislots, channel coefficients followed by the beamforming and SNR may vary over minislots; thus, time-varying utility functions with regard to channel coefficients, beamforming, and SNR should be involved in the utility function design. The following two definitions present the expression of multicast eMBB slice utility and bursty URLLC slice utility, respectively.

\begin{mydef}
For any multicast eMBB slice $s \in {\cal S}^e$, the eMBB utility is defined as the energy efficiency of the RAN slicing system for serving $s$ during the time slot $\bar t$, which is expressed as
\begin{equation}\label{eq:eMBB_revenue}
   \begin{array}{l}
\tilde U_s^e = \sum\limits_{t = 1}^T {U_s^e({{\bm v}_s}(t))} \\
\quad \text{ } = \sum\limits_{t = 1}^T {\sum\limits_{i \in {\cal I}_s^e} {b_s^e(\bar t)\ln \left ( 1 + SNR_{i,s}^e(t) \right )} }  + \\
\qquad \eta \sum\limits_{t = 1}^T { \sum\limits_{j \in {\cal J}} {\left [ E_j - b_s^e(\bar t){\bm v}_{j,s}^{\rm{H}}(t){{\bm v}_{j,s}}(t)\right ]} }, \forall s \in {\cal S}^e
\end{array}
\end{equation}
where ${\bm v}_s(t) = [{\bm v}_{1,s}(t); \ldots; {\bm v}_{J,s}(t)] \in {\mathbb C}^{JK \times 1}$, $\eta$ is a constant energy efficiency coefficient.
\end{mydef}

In this definition, a large SNR may lead to a great achievable data rate. The first term on the RHS can then be considered as the \emph{system profit} for eMBB service provision. The power consumption can be regarded as the \emph{system cost}, and the second term (not multiplied by $\eta$) can be interpreted as the \emph{power balance}.

For URLLC UEs, a high SNR regime may not only improve the approximation accuracy but also reduce the amount of channel uses. Therefore, the system profit for URLLC service provision can be modelled as the summation of SNRs received by all URLLC UEs. Besides, considering that URLLC UEs with more stringent low latency requirements have the priority to be scheduled, latency requirements of URLLC UEs should be involved in the design of the URLLC slice utility function.

\begin{mydef}
For any bursty URLLC slice $s \in {\cal S}^u$, the URLLC utility is defined as the energy efficiency of the RAN slicing system for serving $s$ during the time slot $\bar t$, which is expressed as
\begin{align}\label{eq:URLLC_revenue}
& \tilde U_s^u = \sum\limits_{t = 1}^T {\sum\limits_{i \in {\cal I}_s^u}U_s^u({\bm g_{i,s}}(t))} \nonumber \\
&  \quad \text{ } = \sum\limits_{t = 1}^T {\sum\limits_{i \in {\cal I}_s^u} {\left[ {b_{i,s}^u(t)\ln \left (1 + SNR_{i,s}^u(t) \right ) + }  {\frac{\tilde a{b_{i,s}^u(t)}}{{1 - {e^{ - {D_s}}}}}} \right]}} + \nonumber \allowdisplaybreaks[4] \\
&  \eta \sum\limits_{t=1}^T {\sum\limits_{j \in {\cal J}} \left ({E_j - \sum\limits_{i \in {\cal I}_s^u} {b_{i,s}^u(t)\bm g_{ij,s}^{\rm{H}}(t){\bm g_{ij,s}}(t)} }\right ) }, \forall s \in {\cal S}^u
\end{align}
where ${\bm g}_{i,s}(t) = [{\bm g}_{i1,s}(t);\ldots;{\bm g}_{iJ,s}(t)] \in {\mathbb C}^{JK \times 1}$, $\tilde a$ is a constant.
\end{mydef}

\subsection{Formulated Problem}
Based on the above system models, constraints, and designed utility functions, the RAN slicing problem with a goal of maximizing the total eMBB and URLLC slice utilities during the time slot $\bar t$ can be given by
\begin{subequations}\label{eq:original_problem}
\begin{alignat}{2}
& P0: \mathop {{\rm{maximize}}}\limits_{\scriptstyle \ \; \{ b_{i,s}^u(t), b_s^e(\bar t),\hfill\atop
\scriptstyle  \omega _s^e(\bar t), {\bm v_s}(t), {\bm g_{i,s}}(t)\} \hfill}  \sum\limits_{s \in {{\cal S}^e}} {\tilde U_s^e}  + {\hat \rho} \sum\limits_{s \in {{\cal S}^u}} {\tilde U_s^u} \\
& {\rm subject \text{ } to:} \nonumber \\
& b_s^e(\bar t) \in \{ 0,1\}, \forall s \in {\cal S}^e  \\
& b_{i,s}^u(t) \in \{ 0,1\}, \forall i \in {\cal I}_s^u, s \in {\cal S}^u \\
& SNR_{i,s}^u(t) \ge 5, \forall i \in {\cal I}_s^u, s \in {\cal S}^u \\
& \rm {constraints \text{ } (\ref{eq:eMBB_QoS}),(\ref{eq:RRH_energy}),(\ref{eq:total_bandwidth}) \text{ } are \text{ } satisfied.}
\end{alignat}
\end{subequations}
where $\hat \rho$ is a weight coefficient representing the scheduling priority of inter-slices.

The mitigation of (\ref{eq:original_problem}) is highly challenging. In (\ref{eq:original_problem}), there are two types of variables, i.e., mini-timescale variables $\{b_{i,s}^u(t),{\bm v}_s(t), {\bm g}_{i,s}(t)\}$ and timescale variables $\{b_s^e(\bar t), \omega_s^e(\bar t)\}$, which should be optimized at two different timescales. For $\{b_{i,s}^u(t),{\bm v}_s(t), {\bm g}_{i,s}(t)\}$, they should be optimized at the beginning of each minislot while $\{b_s^e(\bar t), \omega_s^e(\bar t)\}$ should be determined at the beginning of each time slot. This requirement makes (\ref{eq:original_problem}) quite different from some other sequential optimization problems; and thus, some optimization methods cannot be directly applied to mitigate the problem. Additionally, the solution of $\{b_s^e(\bar t), \omega_s^e(\bar t)\}$ needs the acquisition of channel coefficients $\{\bm h_{ij,s}(t)\}$ during the time slot $\bar t$, which may be impossible at the beginning of $\bar t$. A possible proposal of mitigating this difficult problem is to transform it into single timescale problems. After that, some optimization methods can be developed to mitigate single timescale problems. Next, we will discuss how to transform the multi-timescale problem into single timescale problems.

\section{Problem Transformation}
Recall the i.i.d. characteristic of channel coefficients, the objective function of (\ref{eq:original_problem}) (divided by $T$) can be approximated as $\frac{1}{T}\sum\limits_{s \in {{\cal S}^e}} {\tilde U_s^e}  + \frac{1}{T}{\hat \rho} \sum\limits_{s \in {{\cal S}^u}} {\tilde U_s^u} = {\mathbb E}_{\bm {\hat h}}\left [\sum\limits_{s \in {\cal S}^e} \hat U_s^e({\bm {\hat v}}_s) + \hat \rho \sum\limits_{s \in {\cal S}^u} {\sum\limits_{i \in {\cal I}_s^u} \hat U_s^u({\bm {\hat g}}_{i,s})} \right ]$, where ${\bm {\hat h}}$ includes all random channels, and ${\bm {\hat v}}_s$ for all $s \in {\cal S}^e$ and ${\bm {\hat g}}_{i,s}$ for all $i \in {\cal I}_s^u$ and $s\in {\cal S}^u$ are beamformers corresponding to ${\bm {\hat h}}$. Then, by exploiting the SAA technique \cite{kim2015guide}, we can further approximate the expectation of the objective function of (\ref{eq:original_problem}) as its sample average (multiplied by $M$), i.e., $\sum\limits_{s \in {\cal S}^e}{\sum\limits_{m = 1}^M{U_s^e({\bm v}_{sm})}} + \hat \rho \sum\limits_{s \in {\cal S}^u}{\sum\limits_{m =1 }^M{\sum\limits_{i \in {\cal I}_s^u}U_s^u({\bm g}_{i,sm})}}$, where ${\bm v }_{sm}$ and ${\bm g}_{i,sm}$ represent the generated beamformers based on the $m$-th channel coefficient. The convergence of approximating the expectation of a function as its sample average by exploring the SAA technique had been rigorously proven in \cite{kim2015guide}.

Besides, for the probabilistic QoS constraint (\ref{eq:eMBB_QoS}), if the number of samples $M$ is no less than $M^{\star}$ with \cite{liu2016sample,so2013distributionally}
\begin{equation}\label{eq:M_star}
\begin{array}{l}
M^{\star} = \left\lceil \frac{1}{\varepsilon }\left( (N^e + N^u)JK - 1 + \log \frac{1}{\theta } + \right . \right. \\
\qquad \left . \left. \sqrt {2((N^e + N^u)JK - 1)\log \frac{1}{\theta } + {{\log }^2}\frac{1}{\theta }}  \right) \right\rceil
\end{array}
\end{equation}
for any $\theta \in (0, 1)$, then any solution to
\begin{equation}
    \gamma _{i,sm}^e \ge b_{sm}^eR_s, \forall i \in {\cal I}_s^e, s \in {\cal S}^e, m \in {\cal M}:=\{1,\ldots,M\}
\end{equation}
may satisfy (\ref{eq:eMBB_QoS}) with a probability at least $1 - \theta$.

For all $m \in {\cal M}$, if we consider timescale variables $\{\omega_s^e(\bar t), b_s^e(\bar t)\}$ as mini-timescale variables $\{\omega_{sm}^e, b_{sm}^e\}$, then the problem (\ref{eq:original_problem}) can be transformed into $M$ independent single timescale problems that may be mitigated by some optimization methods. Based on the obtained $\{\omega_{sm}^e, b_{sm}^e\}$, an ADMM method \cite{boyd2011distributed} can be exploited to restore $\{\omega_{sm}^e, b_{sm}^e\}$ to $\{\omega_s^e(\bar t), b_s^e(\bar t)\}$.

Next, define ${\bm V}_{sm} = {\bm v}_{sm} {\bm v}_{sm}^{\rm H} \in {\mathbb R}^{JK \times JK}$ for all $s \in {\cal S}^e$, ${\bm G}_{i,sm} = {\bm g}_{i,sm} {\bm g}_{i,sm}^{\rm H} \in {\mathbb R}^{JK \times JK}$, ${\bm H}_{i,sm} = {\bm h}_{i,sm} {\bm h}_{i,sm}^{\rm H} \in {\mathbb R}^{JK \times JK}$ for all $i \in {\cal I}_s^u$ and $s \in {\cal S}^u$, where ${\bm h}_{i,sm} = [{\bm h}_{i1,sm};\ldots;{\bm h}_{iJ,sm}] \in {\mathbb C}^{JK \times 1}$. As ${\rm tr}({\bm {AB}}) = {\rm tr}({\bm {BA}})$ for matrices $\bm A$, $\bm B$ of compatible dimensions, the signal power received at eMBB UE $i$ in $s$ can be expressed as ${\left| {\sum\nolimits_{j \in {\cal J}} {{\bm h}_{ij,sm}^{\rm H}{{\bm v}_{j,sm}}} } \right|^2} = {\left| {{\bm h}_{i,sm}^{\rm H}{{\bm v}_{sm}}} \right|^2} = {\left( {{\bm h}_{i,sm}^{\rm H}{{\bm v}_{sm}}} \right)^{\rm H}}{\bm h}_{i,sm}^{\rm H}{{\bm v}_{sm}} = {\rm tr}({\bm v}_{sm}^{\rm H}{{\bm h}_{i,sm}}{\bm h}_{i,sm}^{\rm H}{{\bm v}_{sm}}) = {\rm tr}({{\bm h}_{i,sm}}{\bm h}_{i,sm}^{\rm H}{{\bm v}_{sm}}{\bm v}_{sm}^{\rm H}) = {\rm tr}({{\bm H}_{i,sm}}{{\bm V}_{sm}})$. Likewise, the signal power received at URLLC UE $i$ in $s$ can be expressed as ${\rm tr}({\bm H}_{i,sm}{\bm G}_{i,sm})$.
The power of each beamforming vector for serving URLLC UEs and eMBB UEs can be written as $\sum\nolimits_{j \in {\cal J}} {{\bm g}_{ij,sm}^{\rm H}{{\bm g}_{ij,sm}}}  = {\bm g}_{i,sm}^{\rm H}{{\bm g}_{i,sm}} = {\rm tr}({\bm g}_{i,sm}^{\rm H}{{\bm g}_{i,sm}}) = {\rm tr}({{\bm g}_{i,sm}}{\bm g}_{i,sm}^{\rm H}) = {\rm tr}({{\bm G}_{i,sm}})$ and $\sum\nolimits_{j \in {\cal J}} {{\bm v}_{j,sm}^{\rm H}{{\bm v}_{j,sm}}}  = {\bm v}_{sm}^{\rm H}{{\bm v}_{sm}} = {\rm tr}({\bm v}_{sm}^{\rm H}{{\bm v}_{sm}}) = {\rm tr}({{\bm v}_{sm}}{\bm v}_{sm}^{\rm H}) = {\rm tr}({{\bm V}_{sm}})$, respectively. Besides, by applying the following property
\begin{equation}\label{eq:matrix_operation}
    \left\{ {\begin{array}{*{20}{l}}
{{\bm V_{sm}} = {\bm v_{sm}}{\bm v}_{sm}^{\rm H} \Leftrightarrow {\bm V_{sm}} \succeq 0,}&{{\rm rank}({\bm V_{sm}}) \le 1}\\
{{{\bm G}_{i,sm}} = {{\bm g}_{i,sm}}{\bm g}_{i,sm}^{\rm H} \Leftrightarrow {{\bm G}_{i,sm}} \succeq 0,}&{{\rm rank}({{\bm G}_{i,sm}}) \le 1}
\end{array}} \right.
\end{equation}
the $m$-th single timescale problem can be formulated as
\begin{subequations}\label{eq:transformed_problem}
\begin{alignat}{2}
& P1:\mathop {{\rm{maximize}}}\limits_{\scriptstyle \ \{ b_{i,sm}^u, b_{sm}^e, \hfill\atop
\scriptstyle  \omega _{sm}^e, {\bm V_{sm}}, {\bm G_{i,sm}}\} \hfill}
\sum\limits_{s \in {{\cal S}^e}} {U_s^e({\bm V_{sm}})}  + \hat \rho \sum\limits_{s \in {{\cal S}^u}} {\sum\limits_{i \in {\cal I}_s^u}{U_s^u({\bm G_{i,sm}})} }  \\
& {\rm subject \text{ } to:} \nonumber \allowdisplaybreaks[4] \\
& \omega _{sm}^e{\log _2}\left( {1 + \frac{{{\rm tr}({{\bm H}_{i,sm}}{\bm V_{sm}})}}{{\phi \sigma _{i,s}^2}}} \right) \ge b_{sm}^eR_s,\forall i \in {\cal I}_s^e, s \in {\cal S}^e \\
& \sum\limits_{s \in {{\cal S}^e}} {b_{sm}^e{\rm tr}({{\bm Z}_j{\bm V}_{sm}})}  + \sum\limits_{s \in {{\cal S}^u}} {\sum\limits_{i \in {\cal I}_s^u} {b_{i,sm}^u{\rm tr}({{\bm Z}_j}{{\bm G}_{i,sm}})} }  \le {E_j},\forall j \\
& \frac{{\rm tr}({\bm H}_{i,sm}{\bm G}_{i,sm})}{\phi \sigma_{i,s}^2} \ge 5b_{i,sm}^u, \forall i \in {\cal I}_s^u, s \in {\cal S}^u  \\
& {\bm V_{sm}} \succeq 0, \forall s \in {\cal S}^e \\
& {{\bm G}_{i,sm}} \succeq 0, \forall i \in {\cal I}_s^u, s \in {\cal S}^u  \\
& {{\rm rank}({\bm V_{sm}}) \le 1}, \forall s \in {\cal S}^e  \\
& {{\rm rank}({{\bm G}_{i,sm}}) \le 1}, \forall i \in {\cal I}_s^u, s \in {\cal S}^u \\
&{\sum\limits_{s \in {{\cal S}^e}} {b_{sm}^e\omega _{sm}^e}  + {W_m^u} \le W} \\
&{b_{sm}^e \in \{ 0,1\} ,\forall s \in {{\cal S}^e}}\\
&{b_{i,sm}^u \in \{ 0,1\} ,\forall i \in {\cal I}_s^u}, s \in {{\cal S}^u}
\end{alignat}
\end{subequations}
where ${\bm Z}_j$ is a square matrix with $J \times J$ blocks, and each block in ${\bm Z}_j$ is a $K \times K$ matrix. In ${\bm Z}_j$, the block in the $j$-th row and $j$-th column is a $K \times K$ identity matrix, and all other blocks are zero matrices.

Although (\ref{eq:transformed_problem}) is a single timescale problem, it is still difficult to mitigate it. First, (\ref{eq:transformed_problem}) simultaneously consists of continuous variables, zero-one variables, and positive semidefinite variables; thus, it is an MISDP problem. Second, the low-rank constraints (\ref{eq:transformed_problem}g) and (\ref{eq:transformed_problem}h) are non-convex. The multiplication of zero-one variables and continuous variables also makes the energy and bandwidth constraints (\ref{eq:transformed_problem}c), (\ref{eq:transformed_problem}i) non-convex. We next discuss how to mitigate this challenging problem.

\section{Problem solution with system generated channels}
In this section, we first try to tackle the non-convex low-rank constraints by exploiting the semidefinite relaxation (SDR) method and then propose a proposal of alternative optimization to mitigate the relaxed problem.

\subsection{Semidefinite Relaxation}
We resort to the SDR method to handling the non-convex low-rank constraints. By directly dropping the low-rank constraints (\ref{eq:transformed_problem}f) and (\ref{eq:transformed_problem}g), we arrive at the following relaxed problem of (\ref{eq:transformed_problem}).
\begin{subequations}\label{eq:relaxed_transformation_problem}
\begin{alignat}{2}
& \mathop {{\rm{maximize}}}\limits_{\scriptstyle \ \{ b_{i,sm}^u, b_{sm}^e, \hfill\atop
\scriptstyle \omega _{sm}^e, {\bm V_{sm}}, {\bm G_{i,sm}}\} \hfill}
\sum\limits_{s \in {{\cal S}^e}} {U_s^e({\bm V_{sm}})}  + \hat \rho \sum\limits_{s \in {{\cal S}^u}} {\sum\limits_{i \in {\cal I}_s^u}{U_s^u({\bm G_{i,sm}})} }  \\
& {\rm subject \text{ } to:} \nonumber \\
& \rm {constraints \text{ } (\ref{eq:transformed_problem}b)-(\ref{eq:transformed_problem}f), (\ref{eq:transformed_problem}i)-(\ref{eq:transformed_problem}k) \text{ } are \text{ } satisfied.}
\end{alignat}
\end{subequations}

Owing to the relaxation, the power matrices $\{{\bm V}_{sm}, {\bm G}_{i,sm}\}$ obtained by mitigating (\ref{eq:relaxed_transformation_problem}) will not meet the low-rank constraints in general. This is because the (convex) feasible set of (\ref{eq:relaxed_transformation_problem}) is a superset of the (nonconvex) feasible set of (\ref{eq:transformed_problem}).
However, if they are, then the eigenvectors will be the optimal solution to (\ref{eq:relaxed_transformation_problem}) and the SDR for ${\bm V}_{sm}$ for all $s \in {\cal S}^e$ and ${\bm G}_{i,sm}$ for all $i \in {\cal I}_s^u$ and $s \in {\cal S}^u$ is tight. If they are not, then the upper bound of the objective function with respect to (w.r.t) the power required by the low-rank transmit beamforming scheme can be obtained. Besides, in this case, we must leverage some methods such as the \emph{randomization/scale} method \cite{ma2010semidefinite} to extract the approximate solution from them.

\subsection{Alternative Optimization}
Considering that the non-convexity of bandwidth and energy constraints is mainly caused by the multiplication of zero-one variables and continuous variables, we propose to decouple these two types of variables and attempt to optimize them alternatively.

\subsubsection{Optimization of zero-one variables}
Given the continuous variables $\{\omega_{sm}^e, {\bm V}_{sm}, {\bm G}_{i,sm}\}$, (\ref{eq:relaxed_transformation_problem}) is reduced to a non-linear integer programming problem due to the existence of the non-linear constraint (\ref{eq:transformed_problem}i) w.r.t $b_{i,sm}^u$ for all $i \in {\cal I}_s^u$ and $s \in {\cal S}^u$. The non-linear constraint significantly increases the difficulty of problem mitigation. Therefore, we try to optimize $b_{sm}^e$ for all $s \in {\cal S}^e$ and $ b_{i,sm}^u$ separately.

\emph{a) Enforcement of eMBB slice requests:}
Given the variables $\{b_{i,sm}^u, \omega_{sm}^e, {\bm V}_{sm}, {\bm G}_{i,sm}\}$, (\ref{eq:relaxed_transformation_problem}) will be reduced to the following problem
\begin{subequations}\label{eq:eMBB_slice_enforcement_problem}
\begin{alignat}{2}
& \mathop {{\rm{maximize}}}\limits_{\{b_{sm}^e\}}
\sum\limits_{s \in {{\cal S}^e}} {U_s^e({\bm V}_{sm})}  \\
& {\rm subject \text{ } to:} \nonumber \\
& \rm {constraints \text{ }  (\ref{eq:transformed_problem}b), (\ref{eq:transformed_problem}c), (\ref{eq:transformed_problem}i), (\ref{eq:transformed_problem}j) \text{ } are \text{ } satisfied.}
\end{alignat}
\end{subequations}

Since all constraints and the objective function in (\ref{eq:eMBB_slice_enforcement_problem}) are linear w.r.t $b_{sm}^e$, (\ref{eq:eMBB_slice_enforcement_problem}) is a linear integer programming problem that can be effectively mitigated by leveraging some standard optimization tools such as MOSEK \cite{MOSEK}.

\emph{b) Optimization of resource mask: }
Given the variables $\{b_{sm}^e, \omega_{sm}^e, {\bm V}_{sm}, {\bm G}_{i,sm}\}$, (\ref{eq:relaxed_transformation_problem}) can be reduced to the following problem
\begin{subequations}\label{eq:resource_mask_problem}
\begin{alignat}{2}
& P2:  \mathop {{\rm{maximize}}}\limits_{\{ b_{i,sm}^u\} } \hat \rho \sum\limits_{s \in {{\cal S}^u}} {\sum\limits_{i \in {\cal I}_s^u} {U_s^u({\bm G}_{i,sm})} } \\
& {\rm subject \text{ } to:} \nonumber \\
& \rm {constraints \text{ } (\ref{eq:transformed_problem}c),(\ref{eq:transformed_problem}d), (\ref{eq:transformed_problem}i), (\ref{eq:transformed_problem}k) \text{ } are \text{ } satisfied.}
\end{alignat}
\end{subequations}

(\ref{eq:resource_mask_problem}) is a non-linear integer programming problem that is hard to be alleviated. In theory, the exhaustive search method can be exploited to obtain the optimal solution to (\ref{eq:resource_mask_problem}). However, the computational complexity of the exhaustive search method is of exponential order.
To reduce the computational complexity, we propose to achieve the resource mask result heuristically. Algorithm 1 depicts the detailed steps of masking resources, the fundamental principle of which is to preferentially allocate resources to URLLC UEs maximizing the URLLC slice utility.
\begin{algorithm}
\caption{Greedy resource mask algorithm, GRM}
\label{alg1}
\begin{algorithmic}[1]
\STATE \textbf{Initialization:} Calculate $U_{s}^u({\bm G}_{i,sm})$ for each $i \in {\cal I}_s^u$, $s \in {\cal S}^u$. Set the vector ${\bm b}^u = {\bm 1}$ and the vector ${ {\bar {\bm b}^u}} = {\bm 0}$.
\WHILE{$\sum\nolimits_{s\in {\cal S}^u} {\sum\nolimits_{i\in {\cal I}_s^u} {\bar b_{i,s}^u} }  \ge 1$}
\STATE $[s^{\star}, i^{\star}] = \arg \mathop {\max }\limits_{i \in {\cal I}_s^u, s \in {\cal S}^u} U_s^u({\bm G}_{i,sm})$
\STATE Check the feasibility of (\ref{eq:transformed_problem}) if ${b_{i^{\star},s^{\star}}^u} = 1$.
\IF {feasible}
\STATE Let ${b_{i^{\star},s^{\star}}^u} = 1$ and ${ {\bar b_{i^{\star},s^{\star}}^u}} = 0$.
\ELSE
\STATE Let ${{\bar b_{i^{\star},s^{\star}}^u}} = 0$.
\ENDIF
\ENDWHILE
\end{algorithmic}
\end{algorithm}

\subsubsection{Optimization of beamforming and bandwidth}
Given the zero-one variables $\{b_{sm}^e, b_{i,sm}^u\}$, the bandwidth and beamforming optimization problem can be formulated as
\begin{subequations}\label{eq:iterative_bandwidth_n_beamforming}
\begin{alignat}{2}
& \mathop {{\rm{maximize}}}\limits_{\{ \omega _{sm}^e, {\bm V_{sm}},{\bm G_{i,sm}}\}}
\sum\limits_{s \in {{\cal S}^e}} {U_s^e({\bm V}_{sm})}  + \hat \rho \sum\limits_{s \in {{\cal S}^u}} {\sum\limits_{i \in {\cal I}_s^u} {U_s^u({\bm G}_{i,sm})} }  \\
& {\rm subject \text{ } to:} \nonumber \\
& \rm {constraints \text{ } (\ref{eq:transformed_problem}b)-(\ref{eq:transformed_problem}f), (\ref{eq:transformed_problem}i) \text{ } are \text{ } satisfied.}
\end{alignat}
\end{subequations}

There are semidefinite matrices and complicate constraints in (\ref{eq:iterative_bandwidth_n_beamforming}), which make it difficult to be optimized. The following Corollary shows how to transform the challenging problem into a standard convex problem such that some convex optimization tools can be exploited to mitigate the problem.
\begin{corollary}\label{corollary_porb}
For all $i \in {\cal I}_s^u$ and $s \in {\cal S}^u$, the probem (\ref{eq:iterative_bandwidth_n_beamforming}) can be equivalently transformed into a standard convex semidefinite programming (SDP) problem via introducing a family of slack variables $\bm f_m = \{f_{i,sm}^u\}$.
\end{corollary}
\begin{proof}
In (\ref{eq:iterative_bandwidth_n_beamforming}), it can be observed that the objective function is linear with regard to ${\bm V}_{sm}$ ($s \in {\cal S}^e$) and ${\bm G}_{i,sm}$. Since the Hessian matrix w.r.t $\omega_{sm}^e$ and ${\bm V}_{sm}$ is negative, the constraint (\ref{eq:transformed_problem}b) is convex. The constraint (\ref{eq:transformed_problem}c) is affine. We next discuss the convexity of (\ref{eq:transformed_problem}i). From (\ref{eq:URLLC_bandwidth}), we can observe that $W_m^u$ is a complicate function of $\bm G_{i,sm}$. However, $W_m^u$ is a quadratic function of $\bm r_m^u$. Therefore, via introducing a family of slack variables $\bm f_m = \{f_{i,sm}^u\}$, we can obtain that (\ref{eq:transformed_problem}i) is equivalent to the following expressions
\begin{equation}\label{eq:slack_equation1}
   {\sum\limits_{s \in {{\cal S}^e}} {b_{sm}^e\omega _{sm}^e}  + \varsigma^{\rm mean}(\bm f_m) + Q^{-1}(\beta)\sqrt{\varsigma^{\rm variance}(\bm f_m)} \le W}
\end{equation}
and
\begin{equation}\label{eq:slack_equation2}
    \begin{array}{l}
f_{i,sm}^u \ge r_{i,sm}^u \\
\quad = \frac{{L_{i,sm}^u}}{{C(SNR_{i,sm}^u)}} + \frac{{{{{Q^{ - 2}}(\alpha)}}}}{{2{{(C(SNR_{i,sm}^u))}^2}}} + \\
 \quad  \frac{{{{{Q^{ - 2}}(\alpha)}}}}{{2{{(C(SNR_{i,sm}^u))}^2}}}\sqrt {1 + \frac{{4L_{i,sm}^uC(SNR_{i,sm}^u)}}{{{{{Q^{ - 2}}(\alpha)}}}}}, \forall i \in {\cal I}_s^u, s \in {\cal S}^u
\end{array}
\end{equation}

By referring to the definition of $\varsigma^{\rm mean}(\cdot)$ and $\varsigma^{\rm variance}(\cdot)$, we can infer that (\ref{eq:slack_equation1}) is a quadratic constraint w.r.t $\{f_{i,sm}^u\}$ and is convex. Besides, from (\ref{eq:slack_equation1}), the following crucial observation can be obtained: $r_{i,sm}^u$ is not only convex but also a monotonically decreasing function over $C(SNR_{i,sm}^u)$. This claim can be obtained by computing the first order and second order derivatives of $r_{i,sm}^u$ over $C(SNR_{i,sm}^u)$. Since $C(SNR_{i,sm}^u)$ monotonically increases with $SNR_{i,sm}^u$ we can further conclude that $r_{i,sm}^u$ monotonically decreases with $SNR_{i,sm}^u$ and then with $\bm G_{i,sm}$. As a result, $r_{i,sm}^u$ is also convex w.r.t $\bm G_{i,sm}$ and (\ref{eq:slack_equation2}) is convex.

By substituting (\ref{eq:slack_equation1}) and (\ref{eq:slack_equation2}) for (\ref{eq:transformed_problem}i), (\ref{eq:iterative_bandwidth_n_beamforming}) can be reformulated as
\begin{subequations}\label{eq:transformed_bandwidth_n_beamforming}
\begin{alignat}{2}
& \mathop {{\rm{maximize}}}\limits_{\{ \omega _{sm}^e, {\bm V_{sm}}, {\bm G_{i,sm}}, {\bm f_m}\}}
\sum\limits_{s \in {{\cal S}^e}} {U_s^e({\bm V}_{sm})}  + \hat \rho \sum\limits_{s \in {{\cal S}^u}} {\sum\limits_{i \in {\cal I}_s^u} {U_s^u({\bm G}_{i,sm})} }  \allowdisplaybreaks[4] \\
& {\rm subject \text{ } to:} \nonumber \\
& \rm {constraints \text{ }  (\ref{eq:transformed_problem}b)-(\ref{eq:transformed_problem}f), (\ref{eq:slack_equation1}),(\ref{eq:slack_equation2}) \text{ } are \text{ } satisfied.}
\end{alignat}
\end{subequations}

Note that if there exists $\bm G_{i,sm}^{\star}$, which is the optimal solution to (\ref{eq:transformed_bandwidth_n_beamforming}), such that the constraint (\ref{eq:slack_equation2}) is satisfied with strict inequality, we can always reduce $f_{i,sm}^{u}$ to make (\ref{eq:slack_equation2}) active without decreasing the objective value of (\ref{eq:transformed_bandwidth_n_beamforming}). Therefore, there is always an optimal solution to (\ref{eq:transformed_bandwidth_n_beamforming}) such that all constraints in (\ref{eq:slack_equation2}) are active. Then, we obtain that (\ref{eq:transformed_bandwidth_n_beamforming}) and (\ref{eq:iterative_bandwidth_n_beamforming}) are equivalent.

From the above analysis, we can know that (\ref{eq:transformed_bandwidth_n_beamforming}) consists of a linear objective function, a quadratic cone constraint of a vector and convex cone constraints of positive semidefinite matrices; hence it is a standard convex SDP problem.
\end{proof}

Semidefinite optimization is a generalization of conic optimization, which allows the utilization of matrix variables belonging to the convex cone of positive semidefinite matrices. Therefore, we can achieve the optimal solution to (\ref{eq:transformed_bandwidth_n_beamforming}) using MOSEK.

\subsection{Restoration of Timescale Variables}
In the above subsections, we investigate the method of obtaining mini-timescale variables $\{\omega_{sm}^e, b_{sm}^e\}$. We next answer the question of \emph{how to restore timescale variables $\{\omega_s^e(\bar t), b_s^e(\bar t)\}$ from $\{\omega_{sm}^e, b_{sm}^e\}$}?

The ADMM method can be exploited to restore timescale variables. A crucial conclusion obtained from the ADMM method is that it will drive mini-timescale variables (or called local variables in ADMM) towards their average value. Particularly, for all $s \in {\cal S}^e$, the continuous variable $\omega_s^e(\bar t)$ can be restored by \cite{boyd2011distributed}
\begin{equation}\label{eq:calculate_wse}
\omega_s^e(\bar t) \approx \frac{1}{M}\sum\limits_{m=1}^M {\omega_{sm}^e}, \forall s \in {\cal S}^e
\end{equation}

However, the average value scheme may not be applicable for the zero-one variable $b_s^e(\bar t)$. To restore $b_s^e(\bar t)$ from $\{b_{sm}^e\}$, a greedy restoration scheme, the fundamental idea of which is to gradually decline eMBB slice requests based on whether the constraints of (\ref{eq:original_problem}) can be satisfied, is developed. The main steps of the scheme are summarized as follows:
\begin{itemize}
    \item Initialization: For all $s \in {\cal S}^e$, initialize $\omega_s^e(\bar t)$ using (\ref{eq:calculate_wse}) and set $C_s^e = \sum\limits_{m = 1}^M{b_{sm}^e}$. For all $i \in {\cal I}_s^u$ and $s \in {\cal S}^u$, let $b_{i,s}^u = 0$. Let $b_s^e(\bar t) = 1$ if the bandwidth $\omega_s^e(\bar t)$ allocated to slice $s$ is non-zero; otherwise, let $b_s^e(\bar t) = 0$.
    \item Decline slice requests: Given $\omega_s^e(\bar t)$ and $b_s^e(\bar t)$, check the feasibility of all $M$ single timescale problems (\ref{eq:transformed_problem}). If all $M$ problems are feasible, then output $b_s^e(\bar t)$ and terminate; otherwise, gradually decline the request of slice $s^{\star} = \mathop {\arg \min }\limits_{s \in {\cal S}^e} C_s^e$. Then, update $\omega_{s^{\star}}^e(\bar t) = 0$ and $b_{s^{\star}}^e(\bar t) = 0$, and repeat to check the feasibility of all $M$ problems.
\end{itemize}

The following Lemma shows the effectiveness of the restoration scheme.

\begin{lemma}\label{lem:restroation_scheme}
For any $\theta \in (0, 1)$, if the number of samples $M \ge M^{\star}$ calculated by (\ref{eq:M_star}), then $b_s^e(\bar t)$ and $\omega_s^e(\bar t)$ for all $s \in {\cal S}^e$ that satisfy the constraints of all $M$ single timescale problems (\ref{eq:transformed_problem}) may be feasible to the problem (\ref{eq:original_problem}) with a probability at least $1 - \theta$.
\end{lemma}
\begin{proof}
(\ref{eq:eMBB_QoS}) is a chance constraint while (\ref{eq:transformed_problem}b) is a robustness constraint. Therefore, the proof of this Lemma may be equivalent to the proof of the feasibility of using a robustness constraint to approximate a chance constraint, which has been rigorously shown in \cite{so2013distributionally} via sample approximations. We omit the detailed proof here as a similar proof can be found in the proof of Theorem 1 in \cite{so2013distributionally}.
\end{proof}

Denote the objective function value of (\ref{eq:transformed_problem}) at the $r$-th iteration by ${\mathscr{F}}\left (b_{sm}^{e(r)}, b_{i,sm}^{u(r)},\omega_{sm}^{e(r)} \right )$, where $\{b_{sm}^{e(r)}$, $b_{i,sm}^{u(r)}$, $\omega_{sm}^{e(r)}\}$ represent the obtained solutions to (\ref{eq:transformed_problem}) at the $r$-th iteration. The main steps of mitigating (\ref{eq:transformed_problem}) can then be summarized in the following Algorithm 2.
\begin{algorithm}
\caption{Iterative acceptance and resource allocation algorithm, IARA-$\alpha\beta$}
\label{alg2}
\begin{algorithmic}[1]
\STATE \textbf{Initialization:} For all $m \in {\cal M}$, randomly initialize $\{\bm G_{i,sm}^{(0)}\}$, $\{\omega_{sm}^{e(0)}\}$, $\{\bm V_{sm}^{(0)}\}$, and let $r_{\rm max} = 250$ and $r = 0$.
\FOR{$m = 1 : M$}
\REPEAT
\STATE Mitigate (\ref{eq:eMBB_slice_enforcement_problem}) to obtain $\{b_{sm}^{e(r+1)}\}$, and call Algorithm \ref{alg1} to generate $\{b_{i,sm}^{u(r+1)}\}$.
\IF{${\mathscr{F}}\left (b_{sm}^{e(r+1)}, b_{i,sm}^{u(r+1)},\omega_{sm}^{e(r)} \right ) \le {\mathscr{F}}\left (b_{sm}^{e(r+1)}, b_{i,sm}^{u(r)},\omega_{sm}^{e(r)} \right )$ }
\STATE Update $b_{i,sm}^{u(r+1)} = b_{i,sm}^{u(r)}$ for all $i \in {\cal I}_s^u$ and $s \in {\cal S}^u$.
\ENDIF
\STATE Mitigate (\ref{eq:transformed_bandwidth_n_beamforming}) to obtain $\{\omega_{sm}^{e(r+1)}\}$, $\{\bm G_{i,sm}^{(r+1)}\}$, and $\{\bm V_{sm}^{(r+1)}\}$.
\STATE Update $r = r + 1$.
\UNTIL{Convergence or reach the maximum number of iteration $r_{\rm max}$, and let $\omega_{sm}^e = \omega_{sm}^{e(r)}$.}
\ENDFOR
\STATE Calculate $\omega_s^e(\bar t)$ using (\ref{eq:calculate_wse}), run the greedy restoration scheme to generate $b_s^e(\bar t)$, and update $\omega_s^e(\bar t) = \omega_s^e(\bar t)b_s^e(\bar t)$ for all $s \in {\cal S}^e$.
\end{algorithmic}
\end{algorithm}

Besides, if we denote $\bm G_{i,sm}^{\star}$ for all $i \in {\cal I}_s^u$ and $s \in {\cal S}^u$ and $\bm V_{sm}^{\star}$ for all $s\in {\cal S}^e$ as the optimal solution to (\ref{eq:transformed_problem}), then the following Lemma presents the effectiveness of applying SDR to (\ref{eq:transformed_problem}) and the convergence of Algorithm \ref{alg2}.
\begin{lemma}\label{lem:SDR}
    \rm {For all $m \in {\cal M}$, the SDR for both ${\bm G}_{i,sm}$ and ${\bm V}_{sm}$ in (\ref{eq:transformed_problem}) is tight, that is,
    \begin{equation}\label{eq:rank_Gis}
    \begin{array}{l}
        {\rm rank}({\bm G}_{i,sm}^{\star}) \le 1, \forall i\in {\cal I}_s^u, s\in {\cal S}^u, \\
        {\rm rank}({\bm V}_{sm}^{\star}) \le 1, \forall s \in {\cal S}^e
        \end{array}
    \end{equation}

    Besides, Algorithm \ref{alg2} is convergent.
    }
\end{lemma}
\begin{proof}
The Lagrangian dual method can be leveraged to prove the tightness of SDR for power matrices. Considering that there are logarithmic terms in the objective function of (\ref{eq:transformed_problem}), we can transform the non-linear objective function into a linear one by introducing slack variables. Then a similar proof can be found in the Appendix of \cite{tang2019service} to prove the tightness of SDR for power matrices.

On the one hand, at the $r+1$-th iteration, we can obtain the optimal solution $b_{sm}^{e(r+1)}$ to (\ref{eq:eMBB_slice_enforcement_problem}) by mitigating a linear integer programming problem. Thus, we have ${\mathscr {F}}\left (b_{sm}^{e(r+1)}, b_{i,sm}^{u(r)}, \omega_{sm}^{e(r)}\right ) \ge {\mathscr {F}}\left (b_{sm}^{e(r)}, b_{i,sm}^{u(r)}, \omega_{sm}^{e(r)}\right )$. From Algorithm \ref{alg2}, we can also conclude that ${\mathscr {F}}\left (b_{sm}^{e(r+1)}, b_{i,sm}^{u(r+1)}, \omega_{sm}^{e(r)}\right ) \ge {\mathscr {F}}\left (b_{sm}^{e(r+1)}, b_{i,sm}^{u(r)}, \omega_{sm}^{e(r)}\right )$. Further, we have ${\mathscr {F}}\left (b_{sm}^{e(r+1)}, b_{i,sm}^{u(r+1)}, \omega_{sm}^{e(r+1)}\right ) \ge {\mathscr {F}}\left (b_{sm}^{e(r+1)}, b_{i,sm}^{u(r+1)}, \omega_{sm}^{e(r)}\right )$ as $\omega_{sm}^{e(r+1)}$ is the optimal solution to the SDP problem (\ref{eq:transformed_bandwidth_n_beamforming}). Therefore, we can conclude ${\mathscr {F}}\left (b_{sm}^{e(r+1)}, b_{i,sm}^{u(r+1)}, \omega_{sm}^{e(r+1)}\right ) \ge {\mathscr {F}}\left (b_{sm}^{e(r)}, b_{i,sm}^{u(r)}, \omega_{sm}^{e(r)}\right )$. On the other hand, in (\ref{eq:transformed_problem}), as the system bandwidth and transmit power are limited, the achievable objective function value is upper-bounded. Then we can say that Algorithm {\ref{alg2}} is convergent.
\end{proof}

\section{Optimization of resource mask and beamforming with sensed channels}
With the system generated channel samples, the above sections obtain the solutions $\{b_s^e(\bar t), \omega_s^e(\bar t)\}$ to (\ref{eq:original_problem}). However, the mini-timescale variables $\{b_{i,s}^u(t), {\bm v}_s(t), {\bm g}_{i,s}(t)\}$ should be optimized based on the sensed channels at each minislot $t$.

At each minislot, with the given timescale variables $\{b_s^e(\bar t), \omega_s^e(\bar t)\}$, the original problem (\ref{eq:original_problem}) can be reduced as
\begin{subequations}\label{eq:mini_time_scale_transformed_problem}
\begin{alignat}{2}
& \mathop {{\rm{maximize}}}\limits_{\{b_{i,s}^u(t),{\bm v}_s(t), {\bm g}_{i,s}(t)\}}
\sum\limits_{s \in {{\cal S}^e}} {U_s^e({\bm v}_s(t))}  +  \hat \rho \sum\limits_{s \in {{\cal S}^u}} {\sum\limits_{i \in {\cal I}_s^u} {U_s^u(\bm g}_{i,s}(t)) }  \\
& {\rm subject \text{ } to:} \nonumber \\
&     \gamma _{i,s}^e(t) \ge b_s^e(\bar t)R_s, \forall i \in {\cal I}_s^e, s \in {\cal S}^e \\
& \rm {constraints \text{ } (\ref{eq:RRH_energy}),(\ref{eq:total_bandwidth}),(\ref{eq:original_problem}c),(\ref{eq:original_problem}d)\text{ } are \text{ } satisfied.}
\end{alignat}
\end{subequations}

According to the analysis presented in the above sections, (\ref{eq:mini_time_scale_transformed_problem}) is a mixed-integer non-convex programming problem with positive semidefinite matrices, which is highly challenging to be mitigated. Therefore, the previously presented SDR method and the proposal of alternative optimization should be leveraged to obtain the solutions $\{b_{i,s}^u(t), {\bm v}_s(t), {\bm g}_{i,s}(t)\}$.

First, given the transmit beamformers ${\bm v}_s(t)$ for all $s \in {\cal S}^e$, ${\bm g}_{i,s}(t)$ for all $i \in {\cal I}_s^u$ and $s \in {\cal S}^u$, (\ref{eq:mini_time_scale_transformed_problem}) will be reduced to a problem similar to (\ref{eq:resource_mask_problem}) with sensed channels at minislot $t$. Then, the proposed GRM algorithm can be utilized to obtain the resource mask $\{b_{i,s}^u(t)\}$.

Second, given the generated resource mask $\{b_{i,s}^u(t)\}$, the problem (\ref{eq:mini_time_scale_transformed_problem}) can be reformulated as a problem similar to (\ref{eq:iterative_bandwidth_n_beamforming}) with a family of constants $\{\omega_s^e(\bar t)\}$ and sensed channels as inputs at $t$. Likewise, after performing the equivalent transformation, MOSEK can be leveraged to obtain the transmit power matrices ${{\bm V}_s}(t)$ and ${\bm G}_{i,s}(t)$.

Recall that the SDR for both $\bm V_{s}(t)$ and $\bm G_{i,s}(t)$ is tight, we therefore can perform the eigenvalue decomposition on $\bm V_{s}(t)$ and $\bm G_{i,s}(t)$ to obtain the optimal beamforming vectors $\bm v_s(t)$ and $\bm g_{i,s}(t)$, respectively.

To sum up, we can depict the logical flow to mitigate the problem (\ref{eq:original_problem}) in Fig. \ref{fig_Algorithm_flow}. At the beginning of each time slot $\bar t$, with the system generated channels the SDRAN-C will follow the flow \textcircled{1} $ \to$ \textcircled{2} $ \to $ \textcircled{3} $\to $ \textcircled{4} $\to$ \textcircled{5} $\to$ \textcircled{6} to achieve $\{b_s^e(\bar t)\}$ and $\{\omega_s^e(\bar t)\}$. With the achieved $\{b_s^e(\bar t)\}$ and $\{\omega_s^e(\bar t)\}$, the SDRAN-C acquires sensed channels with which the following flow \textcircled{3} $\to $ \textcircled{4} $\to$ \textcircled{5} $\to$ \textcircled{6} will be executed to generate beamformers $\{\bm v_s(t)\}$ and $\{\bm g_{i,s}(t)\}$ and resource mask $\{b_{i,s}^u (t)\}$.

\begin{figure}[!t]
\flushleft
\includegraphics[width=2.9in]{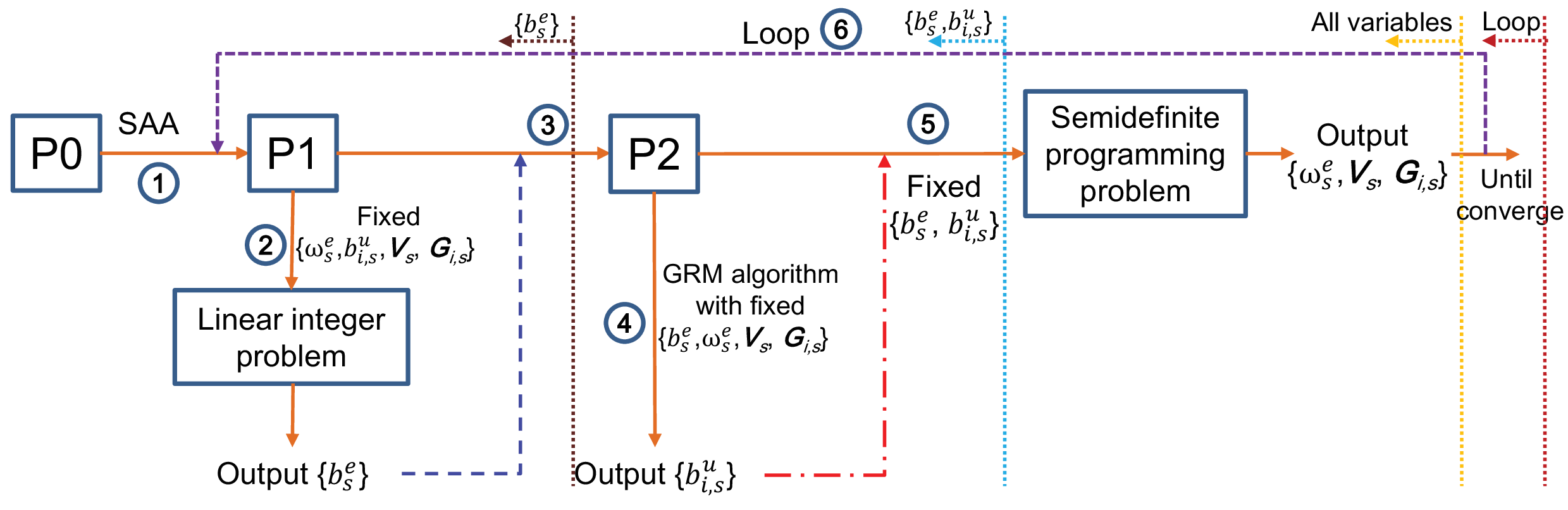}
\caption{Algorithm logical flow.}
\label{fig_Algorithm_flow}
\end{figure}

\section{Prototype of the RAN Slicing System}
In Fig. \ref{fig_RAN_slicing_archi}, we consider a RAN slicing system prototype with four major parties:
\begin{itemize}
    \item \emph{End UEs: }it includes multicast eMBB UEs and bursty URLLC UEs, run their services on the slices managed by the virtualized network slice management.
    \item \emph{Software-defined RAN coordinator (SDRAN-C): } it calculates and updates accepted network slices to accommodate service requirements of end UEs.
    \item \emph{Network slice management (NSM): } a virtualized function aiming to control and manage network slices.
    \item \emph{Core network control function (C$^2$F): } it configures the core network based on the slice requirements.
\end{itemize}

\begin{figure*}[!t]
\centering
\includegraphics[width=5.2in]{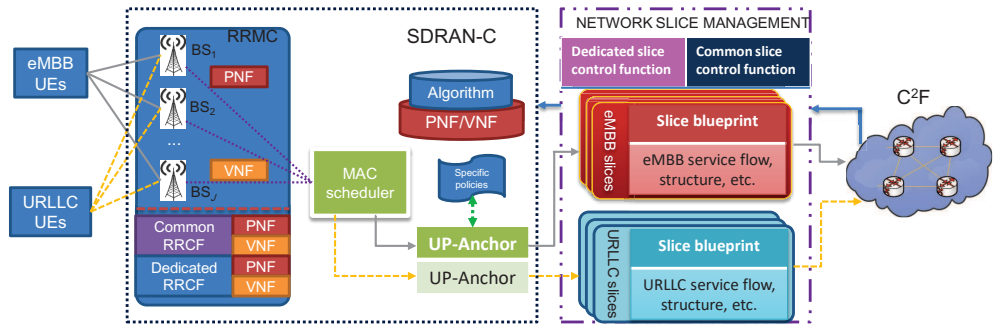}
\caption{A RAN slicing system prototype.}
\label{fig_RAN_slicing_archi}
\end{figure*}

We consider two types of end UEs, i.e., multicast eMBB UEs and bursty URLLC UEs. Each type of UEs possesses several specific features regarding their QoS requirements. These end UEs will send slice requests to SDRAN-C. Upon receiving UEs' slice requests, the SDRAN-C analyzes the slice requirements and makes a decision to accept or decline requests on the basis of the optimal policy.

The SDRAN-C block mainly consists of four components: 1) radio resource manage and control (RRMC); 2) MAC scheduler; 3) user plane anchor (UP-Anchor); 4) Algorithm. If some slices are created, then RRMC is responsible for configuring RAN protocol stacks and QoS according to slice requirements. The slice creation and configuration processes may take up multiple minislots. For example, for slices with high throughput requirements, radio bearers should be configured to support CoMP transmission, and multi-connectivity. For slices with high reliability and low latency requirements, lower frame error rates, reduced RTT, shortened TTI, and/or multi-point diversity schemes are desired to be utilized. Besides, radio resource control function (RRCF) in RRMC, which can be further split into dedicated and common RRCF related functions, will also be activated for UE-specific radio resource control on the basis of virtualized network function (VNF) and/or physical network function (PNF). Particularly, depending on underlying services, RRCF can configure and tailor UP protocol stacks. For example, for slices supporting low latency services, Internet protocol (IP) and related header compression may not be used, and radio link control function may be configured in the transparent mode \cite{rost2017network}.
The MAC scheduler is responsible for traffic scheduling based on the network condition so as to alleviate network congestion. There are many MAC scheduling schemes such as random access, back-off, access class barring. The UP-Anchor is responsible for distributing the traffic according to the configured slice policy, and for encryption
with slice-specific security keys. For example, for slices requested from industry, security, resilience, and reliability of services are of higher priority. Then, policy requirements, e.g., security, resilience, and reliability, should be specified for this type of slices.
Based on slice requests and available resources, the algorithm will periodically calculate, update, and reconfigure accepted slices so as to achieve the maximum total utility. Besides, it is also responsible for the RAN and CN mapping that can be implemented via some configuration protocols \cite{ni2019end}.

The virtualized NSM, which operates on the top of physical and/or virtualized infrastructure, is responsible for creating, activating, maintaining, configuring, and releasing slices during the life cycle of them. Via the dedicated/common slice control function, NSM will generate a network slice blueprint (i.e., a template) for each accepted network slice that describes the structure, configuration, control signals, and service flows for instantiating and controlling the network slice instance of a type of service during its life cycle. The slice instance includes a set of network functions and resources to meet the end-to-end service requirements.

For the C$^2$F, it will interpret the slice blueprint when the slice request arrives. Accordingly, the C$^2$F will arrange the network configuration according to the interpreted blueprint and find the optimal servers and paths to place VNFs to meet the required end-to-end service of the slice.
Besides, C$^2$F possesses multiple data centers for network slicing services. Each data center contains a set of servers (e.g., home subscribe server) with diverse resources, e.g., computing, storage, which are used to support VNFs' services such as identity, independent subscription, session for each network slice. Data centers are connected via backhaul links and can provide services jointly or separately.

\section{Simulation and Performance Evaluation}
This section aims to verify the effectiveness of the proposed algorithm via simulation.

\subsection{Comparison Algorithms and Parameter Setting}
To verify the effectiveness of the proposed IARA-$\alpha\beta$ algorithm, we compare it with three algorithms.
\begin{itemize}
    \item Exhaustive-search-based (ES-$\alpha\beta$) algorithm: Its difference from the IARA-$\alpha\beta$ algorithm is that ES-$\alpha\beta$ searches for $b_{i,s}^u(t)$ for all $s \in {\cal S}^u$ and $i \in {\cal I}_s^u$ via an exhaustive search scheme. The search complexity is $O(2^{S^u})$. The ES-$\alpha\beta$ algorithm follows the following logical flow: at each time slot, execute \textcircled{1} $ \to$ \textcircled{2} $ \to $ \textcircled{3} $\to $ exhaustive search $\to$ \textcircled{5} $\to$ \textcircled{6}; at each minislot, execute \textcircled{3} $\to $ exhaustive search $\to$ \textcircled{5} $\to$ \textcircled{6}.
    \item IARA-$\alpha$ algorithm: It does not take any action to achieve a low URLLC packet blocking probability $\beta$; thus, the system bandwidth reserved for URLLC slices is $W^u(t) = \varsigma^{\rm mean}(\bm r^u(t))$. Similar to IARA-$\alpha\beta$, IARA-$\alpha$ follows the logical flow: at each time slot, execute \textcircled{1} $ \to$ \textcircled{2} $ \to $ \textcircled{3} $\to $ \textcircled{4} $\to$ \textcircled{5} $\to$ \textcircled{6}; at each minislot, execute \textcircled{3} $\to $ \textcircled{4} $\to$ \textcircled{5} $\to$ \textcircled{6}.
    \item IRHS-$\alpha$ algorithm \cite{tang2019service}: IRHS-$\alpha$ algorithm also does not consider the case of achieving a low URLLC packet blocking probability $\beta$. It attempts to accept both eMBB and URLLC slice requests in a greedy way and optimizes beamforming alternatively. Besides, it optimizes the same total slice utility as the other comparison algorithms.
\end{itemize}

The parameter setting of the simulation is summarized as the following: this paper considers a circular area $\mathbb R$ with a radius of $0.5$ km. Three BSs are deployed on the boundary of $\mathbb R$, and the distance among each of them is equal. Multicast eMBB and bursty URLLC UEs are randomly, uniformly, and independently distributed in $\mathbb R$. The transmit antenna gain at each BS is set to be $5$ dB, and a log-normal shadowing path loss model is utilized to calculate the path loss between a BS and a UE. Particular, a downlink path loss is calculated as $H({\rm dB}) = 128.1 + 37.6\log_{10}d$, where $d$ (in km) represents the distance between a UE and a BS. The transmit antenna gain at each BS is set to be $5$ dB, and the log-normal shadowing parameter is set to be $10$ dB \cite{tang2019service}.

Besides, we consider homogeneous BSs with $E_j = 1$ W, $\forall j$, $\lambda_{i,s}= 0.1$ packet per unit time, $\sigma_{i,s}^2 = \sigma^2$, $L_{i,s}^u = L^u$, $\rho_s =  \rho$, $\forall i, s$. Other system configuration parameters are listed in Table \ref{table_simulation_parameters}.

\begin{table}[!t]
\renewcommand{\arraystretch}{1.3}
\caption{System configuration parameters}
\label{table_simulation_parameters}
\centering
\begin{tabular}{|c|c|c|c|c|c|}
\hline
Para. &	Value &  Para. &	Value & Para. &	Value \\\hline
$J$ &	$3$ & $K$ & $2$ & $\eta$ & $100$ \\\hline
$\hat \rho$  & $1$ & $L^u$ & $160$ bits & $\sigma^2$ & $-110$ dBm \\\hline
$T$ & $60$ & $W$ & $10$ MHz & $\kappa$ & $0.032$ \\\hline
$\tilde a$ & $0.1$ & $\alpha$ & $2 \times 10^{-8}$ & $\beta$ & $10^{-6}$ \\\hline
$\epsilon$ & $0.5$ & $\phi$ & $1.5$ & $\theta$ & $0.5$ \\\hline
\end{tabular}
\end{table}

\subsection{Performance Evaluation}
To comprehensively understand the effectiveness of the proposed algorithm, we design the simulation of supporting dedicated service and simulation of providing service multiplexing. In these simulations, the total utility and total power consumption are utilized as evaluation metrics. The total utility is defined as the sum of eMBB slice utility and URLLC slice utility (multiplied by $\hat \rho$), and the total power consumption is defined as the sum of consumed transmit power by all BSs during $T$ minislots.

\subsubsection{Results of dedicated service}
First, we assume that the CoMP-enabled RAN slicing system provides dedicated services for one type of four eMBB UEs, i.e., $S^e = 1$, $I_1^e = 4$, and $S^u = 0$, and are interested in researching the impact of different communication modes, that is, Unicast and multicast, on the system power consumption.
Specifically, we plot the impact of the data rate requirements of eMBB UEs on the power consumption of BSs.
When the data requirement of a UE ranges from $3$ Mb/s to $15$ Mb/s, Fig. \ref{fig_eMBB_single_slice} depicts the trend of the total power consumption obtained by IARA-$\alpha\beta$ using Unicast and multicast communication modes, respectively. In this figure, a zero-one vector $\bm {u}^e$ is leveraged to indicate whether eMBB UEs can be served. If UE $i$ can be served, we let $u_i^e = 1$; otherwise, $u_i^e = 0$ for all $i \in {\cal I}_1^e$.
\begin{figure}[!t]
\centering
\includegraphics[width=2.9in]{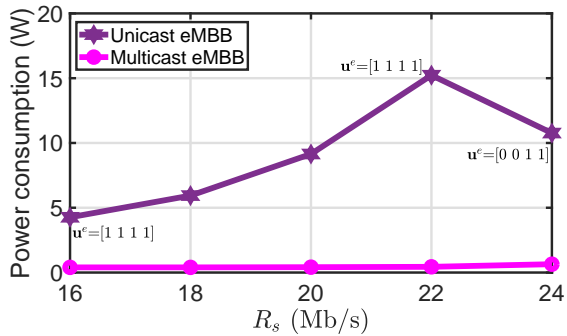}
\caption{eMBB single slice results: power consumption comparison with varying $R_s$.}
\label{fig_eMBB_single_slice}
\end{figure}

We can obtain the following observations from Fig. \ref{fig_eMBB_single_slice}:
\begin{itemize}
    \item Under the similar parameter setting, RRHs selecting the Unicast mode need much greater transmit power than RRHs using the multicast mode. This is because the total system bandwidth can be allocated to each eMBB UE under the multicast mode. When the Unicast mode is adopted, the total system bandwidth will be shared by all eMBB UEs. As a result, greater transmit power is orchestrated for RRHs to satisfy the data requirement of UEs under the Unicast mode. Besides, under the Unicast mode, some UEs cannot be served if a greater data rate is configured. For example, the connections between two UEs and RRHs are interrupted when $R_s = 24$ Mb/s.
    \item Under both Unicast and multicast modes, a larger data rate indicates a greater system power consumption. Nevertheless, the system power consumption under the Unicast mode grows up fast with an increasing $R_s$. The system power consumption under the multicast mode slowly increases with $R_s$.
\end{itemize}

Second, we assume that the CoMP-enabled RAN slicing system supports one bursty URLLC slice i.e., $S^u = 1$ and $S^e = 0$, and interested in studying the difference between schemes of enforcing the minimum SNR constraint and relaxing the minimum SNR constraint. Fig. \ref{fig_URLLC_single_slice} depicts the total utilities obtained by the proposed algorithm using two different SNR schemes.
\begin{figure}[!t]
\centering
\includegraphics[width=2.9in]{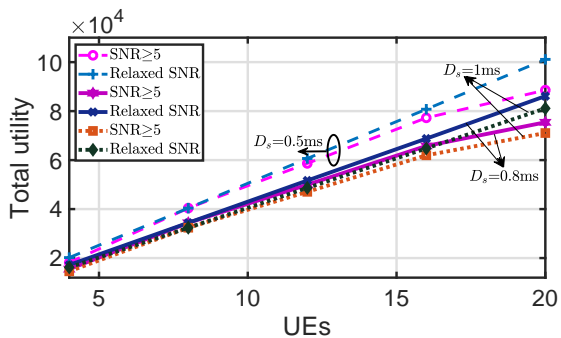}
\caption{URLLC single slice results: total utility comparison with diverse number of UEs.}
\label{fig_URLLC_single_slice}
\end{figure}

From this figure, we can observe that:
\begin{itemize}
    \item As the minimum SNR constraint is relaxed, the feasible region of the RAN slicing problem is enlarged. Accordingly, the total utility achieved by IARA-$\alpha\beta$ under the case of relaxing the minimum SNR constraint is greater than that of enforcing the SNR constraint.
    \item The total utility obtained by the IARA-$\alpha\beta$ algorithm decreases with an increasing $D_s$. This is mainly because a greater $D_s$ leads to a small URLLC profit mapped by $D_s$. Besides, if all UEs can be served, then the total utility of IARA-$\alpha\beta$ increases with an increasing number of UEs.
\end{itemize}

\subsubsection{Results of service multiplexing}
In this subsection, we aim at verifying the effectiveness of the proposed algorithm by comparing it with other algorithms and exploring the impact of some crucial system parameters on the algorithm performance. To this aim, we consider a CoMP-enabled RAN slicing system with a configuration of three multicast eMBB slice requests and two bursty URLLC slice requests. The particular slice request parameters are given in Table \ref{table_slice_request}.
\begin{table}[!t]
\renewcommand{\arraystretch}{1.3}
\caption{Slice request parameters}
\label{table_slice_request}
\centering
\begin{tabular}{|c|c|c|c|c|}
\hline
\multicolumn{3}{|c|}{eMBB slices} &	\multicolumn{2}{|c|}{URLLC slices} \\\hline
$\{I_1^e, R_1\}$ &	$\{I_2^e, R_2\}$ & $\{I_3^e, R_3\}$ & $\{I_1^u, D_1\}$ & $\{I_2^u, D_2\}$ \\\hline
$\{4, 6 {\rm Mb/s}\}$  & $\{6, 4 {\rm Mb/s}\}$ & $\{8, 2 {\rm Mb/s}\}$ & $\{3, 1 {\rm ms}\}$ & $\{5, 2 {\rm ms}\}$ \\\hline
\end{tabular}
\end{table}

As the latency requirement of URLLC UEs is one of the most parameters in URLLC slices, we plot the trend of the total utilities of all comparison algorithms over the latency requirement of URLLC UEs in Fig. \ref{fig_utility_Ds}. In this figure, we reconfigure $\{D_s\}$ of URLLC slices as $D_1 = 0.25 m$ millisecond, $D_2 = 0.5 m$ millisecond with $m \in \{2, 3, \ldots, 10\}$ and keep the parameter configuration of eMBB slices unchanged (as in Table \ref{table_slice_request}). In this figure, we use $E^u$ to denote the power consumption of RRHs for serving URLLC UEs.
\begin{figure}[!t]
\centering
\includegraphics[width=2.9in]{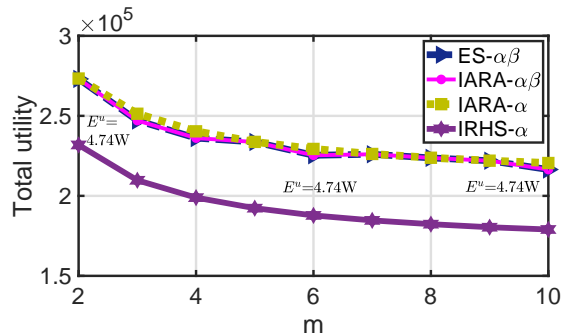}
\caption{Obtained total utilities of all comparison algorithms vs. $D_s$.}
\label{fig_utility_Ds}
\end{figure}

The following observations can be achieved from Fig. \ref{fig_utility_Ds}:
\begin{itemize}
    \item IRHS obtains the smallest total utility. This is because the system bandwidth reserved for eMBB slices are not enough to serve all eMBB UEs although extra system bandwidth is not required in the IRHS algorithm to ensure a low URLLC packet blocking probability. For the other algorithms, they can reserve more system bandwidth for eMBB slices after exploiting the bursty feature of URLLC traffic. Thus, more eMBB UEs can be served, and a greater total utility is obtained.
    \item Although a great $D_s$ might indicate that a small transmit power can be orchestrated for RRHs to satisfy the latency requirement of URLLC UEs, the goal the RAN slicing problem is to provide an energy-efficient communication service. As a result, the varying of $D_s$ does not affect the power consumption of RRHs for serving URLLC UEs.
    \item The total utilities of all comparison algorithms decrease with an increasing $D_s$ as a great $D_s$ generates a small URLLC profit mapped by it.
    \item IARA-$\alpha\beta$ can obtain the same utility as ES-$\alpha\beta$, which may mean that the heuristic resource mask scheme obtains the optimal result.
    \item IARA-$\alpha$ obtains the greatest total utility. Different from IARA-$\alpha\beta$ and ES-$\alpha\beta$, IARA-$\alpha$ suggests the system to orchestrate less system bandwidth for URLLC slices, and then more bandwidth can be allocated to eMBB slices. A greater system bandwidth leads to less power consumption, and thus, a greater utility is obtained.
\end{itemize}

\begin{figure}[!t]
\centering
\includegraphics[width=2.9in]{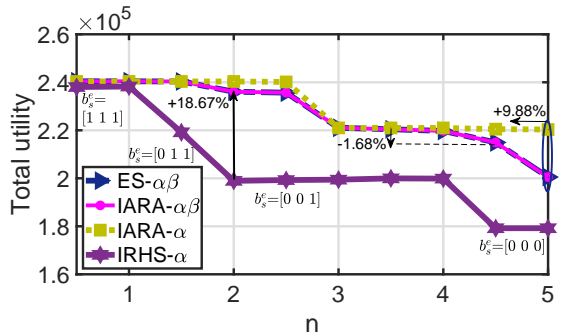}
\caption{Obtained total utilities of all comparison algorithms vs. $R_s$.}
\label{fig_utility_Rs}
\end{figure}

We next plot the trend of total utilities obtained by all comparison algorithms over diverse data rate requirements of eMBB UEs in Fig. \ref{fig_utility_Rs}. In this figure,
we reconfigure $\{R_s\}$ of eMBB slices as $R_1 = 3 n$ Mb/s, $R_2 = 2 n$ Mb/s, and $R_1 = n$ Mb/s with $n \in \{0.5, 1, \ldots, 5\}$ and keep the parameter setting of URLLC slice requirements unchanged (as in Table \ref{table_slice_request}).

The following observations can be achieved from Fig. \ref{fig_utility_Rs}:
\begin{itemize}
    \item IARA-$\alpha$ obtains the greatest total utility. ES-$\alpha\beta$ and IARA-$\alpha\beta$ gain the same total utility that is slightly smaller than that of the IARA-$\alpha$ algorithm when $n < 5$. Besides, compared with IARA-$\alpha\beta$, IARA-$\alpha$ improves the obtained total utility by $9.88\%$ when $n = 5$.
    \item For the IRHS algorithm, it obtains the smallest total utility. When $n \le 1$, the RAN slicing system can accept all three eMBB slice requests. The number of acceptable eMBB slice requests decreases with an increasing $n$. Further, all eMBB slice requests are declined when $n = 5$. For the other algorithms, a great data rate also leads to a reduced number of acceptable eMBB slice requests. This result may be inevitable due to the limitation of system resources. For example, more system bandwidth may be allocated to eMBB UEs when they have higher data rate requirements. However, as the bandwidth orchestrated for eMBB slices is limited, the slice requests of eMBB UEs may be declined when a high data rate is configured.
    \item Owing to the resource constraint, the increase of data rate requirements of eMBB UEs also results in a reduction on the URLLC slice utility. For example, compared with the obtained URLLC slice utility of IARA-$\alpha\beta$ at $n=3.5$, its obtained URLLC slice utility is lowered by $1.68\%$ when $n = 4.5$.
\end{itemize}

Next, we discuss the energy efficiency of all comparison algorithms and plot the obtained total utilities of all comparison algorithms over the energy efficiency coefficient $\eta$ in Fig. \ref{fig_utility_eta}.
From Fig. \ref{fig_utility_eta}, we can observe that:
\begin{figure}[!t]
\centering
\subfigure[Obtained total utility vs. $\eta$]{\includegraphics[width=2.9in]{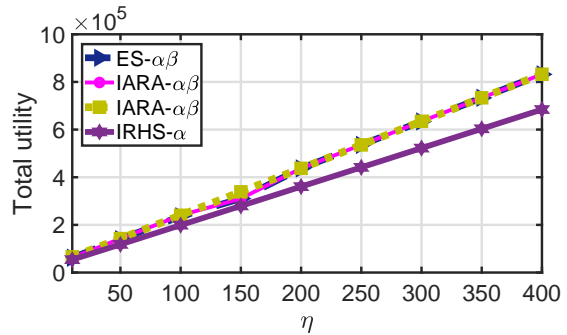}%
\label{fig_first_case}} \\
\subfigure[System power consumption vs. $\eta$]{\includegraphics[width=2.9in]{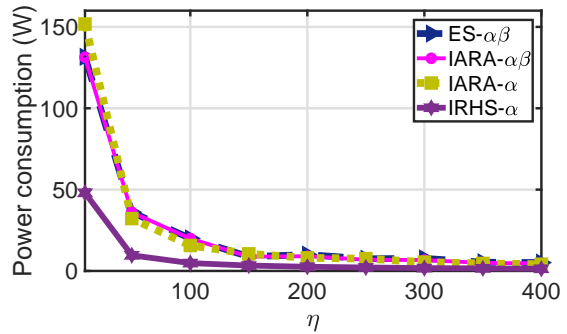}%
\label{fig_second_case}}
\caption{Trends of total utilities and system power consumption of all comparison algorithms.}
\label{fig_utility_eta}
\end{figure}

\begin{itemize}
    \item It is interesting to find that the total utilities obtained by all comparison algorithms rapidly grow up with an increasing energy efficiency coefficient. An increasing $\eta$ leads to a fast decreasing system power consumption and system profit. However, a great $\eta$ will result in a great system power balance that dominates the profit. In other words, this interesting tendency is the outcome of the special mathematical structure of the objective function.
    \item Except for the case of setting $\eta = 0.1$, the obtained total utility by the IRHS algorithm is still smaller than that of the other comparison algorithms. As explained above, the main reason is the lack of system bandwidth reserved for eMBB slices.
\end{itemize}

At last, we plot the trend of the obtained total utilities of all comparison algorithms over different system bandwidth configurations in Fig. \ref{fig_utility_Bandwidth} to verify the impact of system bandwidth on the algorithm performance.
\begin{figure}[!t]
\centering
\includegraphics[width=2.9in]{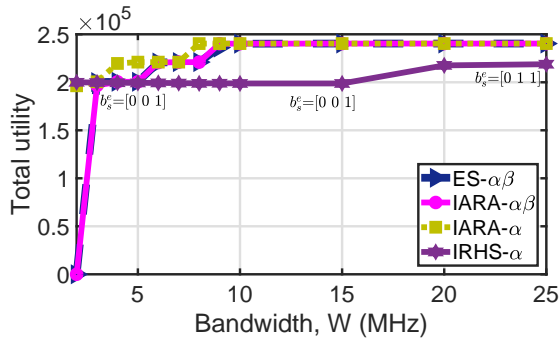}
\caption{Trends of total utilities of all comparison algorithms under different system bandwidth $W$.}
\label{fig_utility_Bandwidth}
\end{figure}
Fig. \ref{fig_utility_Bandwidth} shows that:
\begin{itemize}
    \item Except for the bandwidth configuration of $W = 2$ MHz, the achieved total utilities of the former three algorithms are greater than that of the IRHS algorithm. Owing to the configuration of a small system bandwidth, which is not enough for providing communication services with a low decoding probability and/or a low packet blocking probability requirement, IARA-$\alpha\beta$ and IARA-$\alpha$ cannot accept any slice requests from eMBB and URLLC UEs when $W=2$ MHz.
    \item For the former three comparison algorithms, their achieved total utilities increase with the increase of the system bandwidth. For the IRHS algorithm, its obtained total utility remains unchanged when the system bandwidth varies from $2$ MHz to $15$ MHz. This is because only one eMBB slice can be accepted when $W \le 15$ MHz. When $15 < W \le 25$, the total utility of IRHS is increased as two eMBB slice requests are accepted.
    \item The IARA-$\alpha$ algorithm obtains a greater total utility than that of the IARA-$\alpha\beta$ and ES-$\alpha\beta$ algorithms. At last, it is exciting to find that the greedy-search-based algorithm can optimally mask resources under diverse bandwidth configurations.
\end{itemize}

\section{Conclusion}
This paper investigated the possibility of converging multicast eMBB and bursty URLLC services onto a shared RAN enabling CoMP transmissions. To this aim, we first formulated the CoMP-enabled RAN slicing problem for multicast eMBB and bursty URLLC service multiplexing as a multi-timescale optimization problem with a goal of maximizing the total eMBB and URLLC slice utilities, subject to the total system bandwidth and transmit power constraints. By exploiting some approximation and relaxation schemes, an iterative algorithm with provable performance guarantees was then developed to mitigate the multi-timescale problem. At last, we designed a CoMP-enabled RAN slicing system prototype and conducted extensive simulations to verify the effectiveness of the iterative algorithm.

\appendix

\subsection{Proof of Lemma 1}
Denote $\pi(\bm{b}^u, W^u)$ as the probability that the state of transmissions in a steady system is ${\bm b}^u = \{{\bm b}_1^u, \ldots, {\bm b}_{S^u}^u\}$, where we omit the notation $(t)$ for brevity. According to the standard results from the queueing theory, we have \cite{harchol2013performance}
\begin{equation}\label{eq:pi_b}
\pi(\bm{b}^u, W^u) = G \prod\nolimits_{s = 1}^{S^u} {\left( {\frac{{\rho _s^{{n_s}}}}{{{n_s}!}}} \right)}
\end{equation}
where ${n_s} = \sum\nolimits_{{i} = 1}^{{I}_s^u} {b_{{i,s}}^{u}} $, ${G^{ - 1}} = \sum\nolimits_{{\bar {{\bm b}^u}} \in {\cal B}} {\prod\nolimits_{s = 1}^{S^u} {\left( {\frac{{\rho _s^{{{\bar n}_s}}}}{{{{\bar n}_s}!}}} \right)} }$ with ${\cal B} = \left\{ {\bar {{\bm b}^u}|{\bm \omega}^u{{{\bar {{\bm b}^u}}}^{\rm T}} \le W^u} \right\}$. ${\cal B}$ is the set of all URLLC UE configurations such that the bandwidth constraint related to URLLC traffic is not violated.

Define ${\hat \pi} ({\bm b}^u,W^u)$ for the case when bandwidths and transmission durations are $\hat {\bm \omega}^u $ and $\hat {\bm d}$, respectively, with $q {\rho_s}$ replacing $\rho_s$ in (\ref{eq:pi_b}). Define ${\bm b}_{\backslash s}^u = [{\bm b}_1^u,\ldots,{\bm b}_{s-1}^u,{\bm 0}, {\bm b}_{s+1}^u,\ldots,{\bm b}_{S^u}^u]$, i.e., all UEs of slice $s$ are deactivated. Let $\pi({\bm b}_{\backslash s}^u,W)$ and ${\hat \pi}({\bm b}_{\backslash s}^u,W)$ be the steady probabilities of ${\bm b}_{\backslash s}^u $ under bandwidths ${\bm \omega}^u$ and ${\hat {\bm \omega}}^u$, respectively. According to the standard results of an $M/GI/\infty $ queue system, as $W^u \to \infty $, $\pi({\bm b}_{\backslash s}^u,W^u)$ and ${\hat \pi}({\bm b}_{\backslash s}^u,W^u)$ converge to a common Poisson distribution, i.e.,
\begin{equation}\label{eq:limit_pi}
\begin{array}{l}
\mathop {\lim }\limits_{W^u \to \infty } \pi ({\bm b}_{\backslash s}^u,W^u) = \mathop {\lim }\limits_{W^u \to \infty } \hat \pi ({\bm b}_{\backslash s}^u,W^u)\\
 = \exp \left( { - \sum\nolimits_{l \ne s} {{\rho _l}} } \right)\prod\nolimits_{l \ne s} {\left( {\frac{{\rho _l^{{n_l}}}}{{{n_l}!}}} \right)}
\end{array}
\end{equation}

Recalling the PASTA property, the blocking probability of a new arrival destined to slice $k$ can then be written as
\begin{equation}\label{eq:block_probability}
    p_k({\hat {\bm \omega}^u},{\hat {\bm d}},{\bm \lambda},W^u) = \sum\limits_{{\bm b}^u \in {{\hat {\cal B}}_k}} {\hat \pi ({\bm b}^u,W^u)}
\end{equation}
where ${\hat {\cal B}}_k = \{{\bm b}^u|{\hat {\bm \omega}^u}{{\bm b}^{u{\rm T}}} \le W^u$ and ${\hat {\bm \omega}^u}{{\tilde {\bm b}}^{u{\rm T}}} > W^u\}$ with $\sum\nolimits_{{i} = 1}^{I_k^u} {{ {\tilde b}}_{{i,k}}^{u}} = {n_k} + 1$. ${\hat {\cal B}}_k$ is the set of blocking states of slice $k$. Given ${\bm b} _{\backslash s}^u$, a blocking event must occur when $n_s \in \left\{ {\left\lceil {\frac{{(W^u - {{{\hat {\bm \omega}^u}}}{{\bm b}_{\backslash s}^{u{\rm T}}})}}{{\mathop {\min }\limits_{{i \in {\cal I}_s^u}} { \omega_{{i,s}}^{u}} }/q}} \right\rceil  - \left\lceil {\frac{{{\mathop {\min }\limits_{{i \in {\cal I}_k^u}} { \omega_{{i,k}}^{u}} }}}{{\mathop {\max }\limits_{{i \in {\cal I}_s^u}} { \omega_{{i,s}}^{u}} }/q}} \right\rceil  + 1, \ldots ,\left\lceil {\frac{{(W^u - {{{\hat {\bm \omega}^u}}}{{\bm b}_{\backslash s}^{u{\rm T}}})}}{{ \mathop {\min }\limits_{{i \in {{\cal I}_s^u}}} { \omega_{{i,s}}^{u}} }/q}} \right\rceil } \right\}$. Therefore, using the full probability formula, one can re-write (\ref{eq:block_probability}) as follows
\begin{equation}\label{eq:probability}
    p_k({\hat {\bm \omega}^u},{\hat {\bm d}},{\bm \lambda},W^u) = \sum\limits_{{{\bm b}_{\backslash s}^u} \in {{{\cal B}}_{\backslash s}}} {\varphi ({{\bm b}_{\backslash s}^u},q,{\hat {\bm \omega}}^u,{\hat {\bm d}},W^u) \hat \pi ({{\bm b}_{\backslash s}^u},W^u)}
\end{equation}
where ${{\cal B}_{\backslash s}} = \left\{ {{{\bm b}_{\backslash s}^u}|{{\hat {\bm \omega}^u}}{{\bm b}_{\backslash s}^{u{\rm T}}} \le W^u} \right\}$ and
\begin{equation}\label{eq:probability}
    \varphi ({{\bm b}_{\backslash s}^u},q,{\hat {\bm \omega}^u},{\hat {\bm d}},W^u) = \frac{{\sum\nolimits_{n = n_{lb}}^{\left\lceil {\frac{{(W^u - {{{\hat {\bm \omega}^u}}}{{\bm b}_{\backslash s}^{u{\rm T}}})}}{{ \mathop {\min }\limits_{{i \in {{\cal I}_s^u}}} { \omega_{{i,s}}^{u}} }/q}} \right\rceil} {\frac{{{{(q{\rho _s})}^n}}}{{n!}}} }}{{\sum\nolimits_{n = 0}^{\left\lceil {\frac{{(W^u - {{{\hat {\bm \omega}^u}}}{{\bm b}_{\backslash s}^{u{\rm T}}})}}{{ \mathop {\min }\limits_{{i \in {{\cal I}_s^u}}} { \omega_{{i,s}}^{u}} }/q}} \right\rceil} {\frac{{{{(q{\rho _s})}^n}}}{{n!}}} }}
\end{equation}
where $n_{lb} = { \left\lceil {\frac{{(W^u - {{{\hat {\bm \omega}^u}}}{{\bm b}_{\backslash s}^{u{\rm T}}})}}{{\mathop {\min }\limits_{{i \in {\cal I}_s^u}} { \omega_{{i,s}}^{u}} }/q}} \right\rceil  - \left\lceil {\frac{{{\mathop {\min }\limits_{{i \in {\cal I}_k^u}} { \omega_{{i,k}}^{u}} }}}{{\mathop {\max }\limits_{{i \in {\cal I}_s^u}} { \omega_{{i,s}}^{u}} }/q}} \right\rceil  + 1}$.

For a given $q$, if $\rho_s < 1$, we then have $\varphi ({{\bm b}_{\backslash s}^u},q,{\hat {\bm \omega}^u},{\hat {\bm d}},W^u) < \varphi ({{\bm b}_{\backslash s}^u},1,{\hat {\bm \omega}^u},{\hat {\bm d}},W^u)$ for large $W^u$. According to (\ref{eq:probability}) and (\ref{eq:limit_pi}), we may conclude that $p_k({\bm \omega}^u, {\bm d}, {\bm \lambda}, W^u) \ge p_k({\hat {\bm \omega}^u}, {\hat {\bm d}}, {\bm \lambda}, W^u)$ for all $k \in {\cal S}^u$.

\ifCLASSOPTIONcaptionsoff
  \newpage
\fi




%
\bibliographystyle{IEEEtran}
\bibliography{Network_slicing}

\end{document}